\RequirePackage{fix-cm}
\documentclass[smallextended]{svjour3}       
\smartqed  
\usepackage{graphicx}
\usepackage{algorithm}
\usepackage{tikz}

\usepackage[T1]{fontenc}
\usepackage{textcomp}
\usepackage[utf8]{inputenc}
\usepackage[english]{babel}
\usepackage{refstyle}
\usepackage{float}
\usepackage{mathtools}
\usepackage{amsmath, amssymb}
\usepackage{cancel}
\usepackage{xargs}[2008/03/08]
\usepackage[pdfusetitle,
 bookmarks=true,bookmarksnumbered=false,bookmarksopen=false,
 breaklinks=false,pdfborder={0 0 1},backref=false,colorlinks=false]
 {hyperref}

\usepackage{thm-restate}

\makeatletter

\newtheorem{assertion}{Incorrect Theorem}

\floatstyle{ruled}
\newfloat{algorithm}{tbp}{loa}
\floatname{algorithm}{Algorithm}

\usepackage{algpseudocode}
\newcommand\ALG@originalstep{\ALG@step}
\newcommand\ALG@alternativestep{$\triangleright$}

\algblock{Input}{EndInput}
\algnotext{EndInput}
\algrenewtext{Input}{\textbf{Input}}
\newenvironment{inputblock}{
    \renewcommand\ALG@step{\ALG@alternativestep}
    \Input
}{
    \EndInput
    \renewcommand\ALG@step{\ALG@originalstep}
}
\algblock{Output}{EndOutput}
\algnotext{EndOutput}
\algrenewtext{Output}{\textbf{Output}}
\newenvironment{outputblock}{
    \renewcommand\ALG@step{\ALG@alternativestep}
    \Output
}{
    \EndOutput
    \renewcommand\ALG@step{\ALG@originalstep}
}
\newcommand{\returnstmt}[1]{\State \textbf{return} #1}

\usepackage{tikz-cd}
\makeatother

\newcommand{\abs}[1]{\left| #1 \right|}
\newcommand{\virgolette}[1]{``#1''} 

\newcommandx\evaltime[1][usedefault, addprefix=\global, 1=]{\mathord{\mathrm{EO}}_{#1}}%
\newcommandx\oracletime[1][usedefault, addprefix=\global, 1=]{\mathord{\mathrm{EO}}_{#1}}%
\newcommandx\submintime[2][usedefault, addprefix=\global, 1=, 2=k]{\mathord{\mathrm{\#\#\#}}_{#1,#2}}%

\newcommandx\subrettime[2][usedefault, addprefix=\global, 1=, 2=k]{\mathord{\mathrm{SFM}}_{#1}}%
\newcommandx\supaddcovertime[2][usedefault, addprefix=\global, 1=, 2=k]{\mathord{\mathrm{SAC}}_{#1,#2}}%
\newcommandx\leastcoretime[1][usedefault, addprefix=\global, 1=]{\mathord{\mathrm{LC}}_{#1}}%
\global\long\def\bigo{\mathop{O}}%
\global\long\def\bigo{\mathop{O}}%

\global\long\def\powerset#1{2^{#1}}%

\global\long\def\setdef{\;\middle|\;}%

\global\long\def\argmax{\mathop{\arg\max}}%
\global\long\def\argmin{\mathop{\arg\min}}%

\global\long\def\algset{\leftarrow}%

\global\long\def\supaddcover#1{\overline{#1}}%

\begin{document}

\title{Faster Algorithms for the Least-Core value and the Nucleolus in Convex Games}


\author{Giacomo Maggiorano
\and Alessandro Sosso
\and Gautier Stauffer
}



\institute{
Giacomo Maggiorano 
\at HEC Lausanne, University of Lausanne, Switzerland
\email{giacomo.maggiorano@unil.ch}
ORCID: 0009-0005-9790-4616
\and 
Alessandro Sosso
\at Department of Computer Science, Aarhus University, Aarhus, Denmark
\email{sosso@cs.au.dk}
ORCID: 0009-0005-5136-3391
\and 
Gautier Stauffer
\at HEC Lausanne, University of Lausanne, Switzerland
\email{gautier.stauffer@unil.ch}
ORCID: 0000-0002-2809-8447
}

\date{}

\maketitle

\begin{abstract}
The nucleolus is a central solution concept in cooperative game theory. While its computation is NP-hard in general, it can be computed in polynomial time for convex games; however, the only published polynomial-time algorithm relies on the ellipsoid method. We develop a combinatorial alternative based on reduced games and iterative least-core value computations. Leveraging submodular function minimization and polyhedral structure in a novel way, we obtain a faster combinatorial algorithm for computing the least-core value, improving the oracle complexity by a factor $n^3$ over previous approaches. As a consequence, we obtain a new strongly polynomial-time and combinatorial algorithm for computing the nucleolus in convex games. Preliminary analysis indicates an improved oracle complexity compared to the ellipsoid-based algorithm.
\end{abstract}

\newpage
\section{Introduction}

Cooperative game theory provides a powerful framework for modeling situations in which independent agents can achieve mutual benefits through collaboration. Its applications span a wide range of fields, including supply chain management~\cite{lozano2013cooperative}, transportation planning~\cite{guajardo2016review}, power systems~\cite{churkin2021review}, water resource allocation~\cite{wang2003water}, inventory management~\cite{fiestras2011cooperative}, maritime spatial planning~\cite{kyriazi2017cooperative}, and kidney exchange~\cite{benedek2025partitioned}. For a recent survey on applications in operations management, we refer the reader to~\cite{luo2022core}.

A central objective in cooperative game theory is to determine when collaboration among agents is beneficial and to design equitable mechanisms---known as \emph{solution concepts}---for distributing the resulting collective gains. Prominent examples include the \emph{core}~\cite{gillies1953some,shapley1951value}, the \emph{Shapley value}~\cite{shapley1953value}, and the {\em least-core} and the \emph{nucleolus}~\cite{schmeidler1969nucleolus}. The present paper concentrates on the latter two concepts.

While conceptually appealing, the nucleolus poses significant computational challenges as the number of agents increases. Although it can be computed efficiently for small instances~\cite{benedek2021finding}, computing the nucleolus is NP-hard in general. This hardness result holds even for several important classes of games, including flow games~\cite{deng1999algorithmic}, minimum-cost spanning tree games~\cite{faigle1998note}, weighted threshold games~\cite{elkind2007computational}, and $b$-matching games~\cite{konemann2024complexity}. Nonetheless, polynomial-time algorithms have been developed for several structured settings, including tree games~\cite{megiddo1978computational,granot1996kernel}, assignment games~\cite{solymosi1994algorithm}, matching games~\cite{kern2003matching,konemann2020computing}, shortest path games~\cite{baiou2019algorithm}, and other classes amenable to dynamic programming~\cite{koenemann2023framework} (see also~\cite{granot1998characterization,chen2012computing,potters2006nucleolus,baiou2022network,xiao2023arboricity,baiou2022horizontal} for additional examples).

Convex games~\cite{shapley1971cores} form a foundational subclass of cooperative games, characterized by the property that a player's marginal contribution increases as the coalition grows. Key examples include airport games~\cite{littlechild1976further,hou2013convexity}, bankruptcy games~\cite{aumann1985game,curiel1987bankruptcy}, standard tree games~\cite{granot1996kernel}, and sequencing games~\cite{curiel1989sequencing}. The nucleolus of a convex game can be computed in (oracle) polynomial time, assuming access to an evaluation oracle for the value function, using an ellipsoid-based algorithm~\cite{faigle2001computation}.

Structural results~\cite{maschler1971kernel} suggest an alternative approach for computing the nucleolus based on iteratively solving subgames involving fewer players, known as \emph{reduced games}. A naive implementation of this reduced-game approach leads to an exponential-time algorithm. An established yet unpublished manuscript~\cite{kuipersPolynomialTimeAlgorithm1996} claimed a polynomial-time implementation. 

In this paper, we revisit this line of research, identify a critical flaw in the approach of~\cite{kuipersPolynomialTimeAlgorithm1996}, and develop a corrected combinatorial framework for computing the nucleolus in convex games. Our main contributions are threefold. First, we provide a counterexample to the approach of~\cite{kuipersPolynomialTimeAlgorithm1996} and establish a valid reduced-game--based procedure. Second, we introduce a new strongly polynomial-time combinatorial algorithm for computing the least core value by reformulating the feasibility of a candidate least-core value using Frank’s discrete sandwich theorem and submodular function minimization. Third, we show that this improvement leads to a new strongly polynomial-time combinatorial algorithm for computing the nucleolus in convex games, with improved oracle complexity compared with existing approaches, according to our preliminary analysis.

In addition, our techniques yield a new algorithm for testing the non-emptiness of base polytopes associated with crossing submodular functions over families that are closed under complement, which may be of independent interest.

\subsubsection*{Organization of the paper}

In Section~\ref{sec:notions}, we present the basic definitions and key results on cooperative games. In Section~\ref{sec:algoRG}, we exhibit a counterexample to an incorrect theorem in~\cite{kuipersPolynomialTimeAlgorithm1996}, thereby demonstrating the flaw in the proposed approach, and we provide a corrected reduced-game procedure that computes the nucleolus in convex games in polynomial time.

The reduced-game approach relies on the iterative resolution of least-core problems, which can be carried out either via the ellipsoid method or through the combinatorial approach presented in~\cite{kuipersPolynomialTimeAlgorithm1996}, itself building on Fujishige’s non-emptiness test for base polytopes associated with crossing submodular functions~\cite{fujishige1984structures,kuipersPolynomialTimeAlgorithm1996,schrijverCombinatorialOptimizationPolyhedra2004,frank2011connections}. In Section~\ref{sec:leastcore}, we present a faster combinatorial method for computing the least-core value. Our approach reformulates the feasibility of a candidate least-core value through Frank’s discrete sandwich theorem~\cite{frankAlgorithmSubmodularFunctions1982a}, combined with submodular function minimization, leading to a simpler structural characterization and improved complexity bounds.

Finally, in Section~\ref{sec:complexity}, we compare our least-core value algorithm with the alternative approach proposed in~\cite{kuipersPolynomialTimeAlgorithm1996} and show that our method improves the running time by a factor of~$n^{3}$. We then analyze the complexity of our new algorithm for computing the nucleolus in convex games. Our preliminary analysis indicates that its oracle complexity improves upon that of the previous algorithm of~\cite{faigle2001computation}. The $n^{3}$ improvement in least-core value computation plays a key role in this result.

\subsection*{Remark}

We further show that one of the algorithmic results from Section~\ref{sec:leastcore} extends to testing the non-emptiness of base polytopes associated with crossing submodular functions over families that are \emph{closed under complement}. This yields a faster procedure than the one derived from Fujishige’s framework~\cite{fujishige1984structures,kuipersPolynomialTimeAlgorithm1996,schrijverCombinatorialOptimizationPolyhedra2004,frank2011connections} in the general case. To the best of our knowledge, this extension is new and may be of independent interest. For readability, we present the corresponding self-contained generalization in Appendix~\ref{app:C}.

\section{Definitions and preliminaries}\label{sec:notions}
{\it In what follows, we adopt the notation of Peleg and Sudh\"olter~\cite{peleg2007introduction};  proofs of standard results omitted here may be found in their book.}
\vspace{1ex}
\global\long\def\core{\mathop{\mathrm{Core}}}%
\global\long\def\leastcore{\mathop{\mathrm{LeastCore}}}%
\global\long\def\excess{\varepsilon}%
\global\long\def\optexcess{\excess^{*}}%
\newcommandx\excessgame[1][usedefault, addprefix=\global, 1=\excess]{v_{#1}}%

A \emph{(profit) cooperative game with transferable utility} (simply called a \emph{game} hereafter) is a pair $(N,v)$, where $N$
is a non-empty set of $n$ \emph{players},  representing the agents that can decide to cooperate, while $v\colon\powerset N\rightarrow\mathbb{R}$, with $v(\emptyset)=0$, is the \emph{value  function} of the game. 
The value $v(S)$ can be thought of as the utility (money in many real-world settings) that the players in $S$ can achieve through cooperation. 
Each subset of players $S \subseteq N$ is called a \emph{coalition}. Sometimes we assume without loss of generality that $N=\{1,...,n\}$.

In this paper, we are interested  in games where the marginal value brought by a player increases as the coalition grows.
Formally, the value function $v$ satisfies  $v(S\cup\{i\})-v(S) \leq v(T\cup\{i\})-v(T)$ for all $i\in N$ and $S\subseteq T\subseteq N\setminus\{i\}$, or equivalently, the following holds for all  $S,T\subseteq N$:
\begin{equation}\label{eq:supermodularity}
v\left(S\right)+v\left(T\right)\leq v\left(S\cap T\right)+v\left(S\cup T\right). 
\end{equation}

A set function $v\colon\powerset N\rightarrow\mathbb{R}$ that satisfies the previous property is said to be {\em supermodular}, while (profit) games with a value function that is supermodular are called {\em convex} \cite{shapley1971cores}. 
A set function $v\colon\powerset N\rightarrow\mathbb{R}$ is instead 
said to be \emph{submodular} if $-v$ is supermodular.

In the remainder of the text, we will also use
two other classes of set functions, that satisfy weaker forms of supermodularity. A set function $v\colon\powerset N\rightarrow\mathbb{R}$ is said
to be \emph{intersecting~supermodular} if it satisfies (\ref{eq:supermodularity})
for all $S,T\subseteq N$ such that $S\cap T\neq\emptyset$ and it is said to be {\emph{crossing~supermodular}} if it satisfies (\ref{eq:supermodularity})
for all $S,T\subseteq N$ such that $S\cap T\neq\emptyset$ and $S\cup T\neq N$. Games with intersecting and crossing supermodular value  functions were referenced as {\em near-convex} and {\em pseudo-convex} games in \cite{kuipersPolynomialTimeAlgorithm1996}, but, for clarity, we will refrain from employing this notation in the present work.
Sometimes, we consider the restriction of a set function $v\colon\powerset N\rightarrow\mathbb{R}$  to $2^U$ for some $U\subsetneq N$. We say that the function is intersecting (resp.~crossing) supermodular {\em on $U$} to express the fact that property (\ref{eq:supermodularity}) holds for all $S,T\subseteq U$ such that $S\cap T\neq \emptyset$ (resp.~$S\cap T\neq \emptyset$ and $S\cup T\neq U$).

Convex games are \emph{superadditive}, that is, $v\left(S\right)+v\left(T\right)\leq v\left(S\cup T\right)$ for all $S,T\subseteq N$ with $S\cap T=\emptyset$. 
In particular, the grand coalition $N$ maximizes the total value that can be generated by any family of coalitions partitioning it. Hence, $v(N)$ represents the maximum utility available to the players and is referred to as the \emph{value of the game}. One of the central goals of (superadditive) cooperative game theory is to provide sound \emph{solution concepts} to distribute the benefits of cooperation, that is, this value, among the players. This problem is of major importance, as in real-world scenarios the presence of a sharing rule that is considered \virgolette{fair} by the agents involved has been proven to be a key element for fostering collaboration \cite{bouncken2020value}.
It is easy to distribute $v(N)$ to the different players to guarantee {\em individual rationality} (i.e.\ each player gets at least her individual value)  by splitting the surplus evenly and thus giving each player $i$ the share $v(\{i\}) + 1/n \cdot (v(N)-\sum_{j\in N} v(\{j\}))$. 
However, this simple sharing rule might not be fully satisfactory as some subsets of players might have an incentive to leave the grand coalition if what they receive is less than the value they could get together. 
When no coalition has an incentive to leave $N$, the sharing mechanism is said to be {\em stable}. 

The \emph{core} of a game \cite{gillies1953some} has been established as one of the most relevant solution concepts as it formalizes the aforementioned notion of stability.  We define a \emph{preimputation} as any vector $x \in \mathbb{R}^n$ such that $x(N)=v(N)$, where we use the notation $x(S)=\sum_{i \in S}x_i$ for $S \subseteq N$. 
A preimputation can be seen as a possible sharing of the total gain obtained when all players cooperate: $x_i$ is the share of $v(N)$ received by player $i$. 
Given a game $(N,v)$, the \emph{core} of $(N,v)$ is the set of stable preimputations:
\begin{equation*}
	\core(N,v) \coloneqq \left\{ x \in \mathbb{R}^n \setdef x(S) \ge v(S) \quad \forall S \subseteq N, \quad x(N)=v(N) \right\}.
\end{equation*}

We note that the core of a game is a polyhedron, and in the general case may be empty. For convex games, the core is closely related to the base polytope of an extended polymatroid and it has the property of always being non-empty \cite{shapley1971cores}.
In particular, an extreme point can be computed with the greedy algorithm by setting, for $i=1,...,n$, $x_{i} \algset v( \{1,...,i\}) - v(\{1,...,i-1\})$. 
Note that when dealing with computational aspects of games, we usually assume that the value function $v$ is given through an evaluation oracle, denoted as $\oracletime[v]$, and we analyze algorithmic complexity in terms of the number of calls to this evaluation oracle plus possibly additional operations, as a function of the number of players $n$. 
For instance, the greedy algorithm has complexity $O(n\cdot \oracletime[v])$.

The {\em prenucleolus} is a single point solution concept \cite{schmeidler1969nucleolus}  defined as follows.  Given a game $(N,v)$ and a preimputation $x\in \mathbb{R}^n$, let us denote by $\theta(x,v) \in \mathbb{R}^{2^n-2}$ the vector whose components are $x(S)-v(S)$ $\forall S \subsetneq N$, $S\neq \emptyset$ sorted in non-decreasing order.
The \emph{prenucleolus} $\eta(N,v)$ of $(N,v)$ is the (unique) preimputation that lexicographically maximizes $\theta(x,v)$ over the set of preimputations, i.e.\
\begin{equation*}
\eta(N,v)\coloneqq \mathop{\arg \mathrm{lex} \max_{x\in\mathbb{R}^{n}}} \left\{\theta(x,v) \setdef x(N)=v(N)\right\}.
\end{equation*}

The prenucleolus is contained in the {\em least core}, which consists of the set of  preimputations that maximize the lowest value $\optexcess$ of $x(S)-v(S)$ for $S \subsetneq N$, $S \neq \emptyset$ over all possible preimputations $x$ (i.e., that minimize the largest ``dissatisfaction''). 
The value $\optexcess$ is referred to
as the \emph{least core value}. 
It can be computed through the following linear program: 
\begin{equation}\label{eq:least-core}
\begin{alignedat}{3}
 \optexcess = \ & \text{maximize} & \excess\\ 
 & \text{subject to}\qquad & x\left(S\right) & \geq v\left(S\right)+\excess & \qquad &\forall S\subsetneq N,\; S\neq\emptyset \\
 &  & x\left(N\right) & =v\left(N\right), 
\end{alignedat}
\end{equation}
while the least core is formally defined as:
\[
\leastcore\left(N,v\right)\coloneqq\left\{ x\in\mathbb{R}^{n}\setdef\left(x,\optexcess\right)\text{ is an optimal solution to (\ref{eq:least-core})}\right\} .
\]

Problem (\ref{eq:least-core}) is always feasible and, for $n \ge 2$, bounded; thus, the least core is never empty. 
When the core is non-empty, we have by definition that $\optexcess \geq 0$ and $\leastcore(N,v) \subseteq \core(N,v)$, implying that the prenucleolus belongs to the core. The \emph{nucleolus} is defined analogously to the prenucleolus, with the maximization restricted to individually rational preimputations. 
When the core is non-empty---as is the case, for example, of convex games---prenucleolus and nucleolus coincide~\cite{peleg2007introduction}. Therefore, in what follows, we use the notation $\eta(N,v)$ to denote both concepts.

\begin{remark}
In the following, when $N$ is clear from the context, we will slightly abuse notation and use $\core(v)$ (resp.~$\eta(v)$, $\leastcore(v)$), instead of $\core(N,v)$ (resp.~$\eta(N,v)$, $\leastcore(N,v)$).
\end{remark}

Identifying a point in the least core is NP-hard in general \cite{faigle2000note}. Nonetheless, for convex games, problem (\ref{eq:least-core}) can be solved in (oracle) polynomial time using the ellipsoid method \cite{faigle2001computation}. 
This efficiency arises because the associated separation problem --- verifying whether a given solution $(\bar{x}, \bar{\excess})$ satisfies $\bar{x}(S) \geq v(S) + \bar{\excess}$ for every non-empty $S \subsetneq N$, and identifying a violated inequality if the condition fails --- is polynomially solvable.
Indeed, maximizing a supermodular function (equivalently, minimizing a submodular function) can be performed in polynomial time \cite{grotschel1981ellipsoid} and even in strongly polynomial time \cite{iwata2001combinatorial}.
Such algorithms can effectively minimize the submodular function $\bar{x}(S)-v(S)$ over all subsets $S$ containing a particular element $i$ and excluding another element $j$, for arbitrary $i,j \in N$.
Consequently, it suffices to verify whether the minimum value computed over all possible pairs $(i,j)$ is at least $\bar{\excess}$ (and provide a minimizing subset $S$ if this condition is violated).

The dual of problem (\ref{eq:least-core}) is
\begin{equation}\label{eq:least-core-dual}
\begin{alignedat}{3}
 & \text{minimize} & \sum_{S\subsetneq N, S\neq\emptyset}\mu_{S}\cdot v\left(S\right) & +\mu_{N}\cdot v\left(N\right)\\ 
 & \text{subject to}\qquad & \mu_{S} & \leq0 & \quad & \forall S\subsetneq N,\; S\neq\emptyset \\
 &  & \sum_{ S\subsetneq N:s \in S}\mu_{S}+\mu_{N} & =0 & \quad & \forall s\in N \\
 &  & \sum_{S\subsetneq N, S\neq\emptyset}\mu_{S} & =-1. 
\end{alignedat}
\end{equation}
Since (\ref{eq:least-core})  and  (\ref{eq:least-core-dual}) are both feasible, the optimal value of (\ref{eq:least-core-dual}) is equal to
the least core value $\optexcess$. 
Note also that, in an optimal dual solution, there is always a non-empty set $S\subsetneq N$ such that $\mu_S<0$. Now, by complementary slackness, for such $S$, any optimal solution $(x, \varepsilon^*)$ to (\ref{eq:least-core}) satisfies $x(S)=v(S)+\optexcess$. We call such a coalition $S$ an {\em essential coalition} of the game $(N,v)$. 
For convex games, it is also possible to retrieve such an essential coalition in polynomial time through the ellipsoid method \cite{faigle2001computation}. 
In Section~\ref{sec:algoRG}, we demonstrate how to exploit this property inductively to compute the (pre)nucleolus in convex games, utilizing the concept of \emph{reduced game}, as introduced in \cite{davis1965kernel}. We now present some fundamental properties of these games.

Reduced games allow to reduce the computation of the prenucleolus to a smaller subset of players once some components are known for the other players.
Assume that  we are given a game $(N,v)$ and we know the components $\eta_j$ of its prenucleolus $\eta$ for all $j\in N\setminus S$ for some $ S\subsetneq N$, $S \neq \emptyset$.  Consider the game $(S,v_{S,\eta})$ with: 
	\begin{equation*}
		v_{S, \eta}(T)=\begin{cases}
			0 \qquad &\text{if }T=\emptyset\\
			v(N)-\eta(N \setminus S) \qquad &\text{if }T=S\\
			\max_{R \subseteq N\setminus S}\Big(v(T \cup R)-\eta(R)\Big)\qquad \qquad &\forall T \subsetneq S, T \neq \emptyset.
		\end{cases}
	\end{equation*}
The game $(S, v_{S, \eta})$ is the \emph{reduced game} with respect to $S$ and $\eta$ (note that a reduced game can also be defined with respect to a preimputation other than $\eta$). 
Whenever $\eta$ is in the core of the original game $(N,v)$ (that is, when the core is non-empty and the prenucleolus and nucleolus coincide,  as in convex games), we can give a simpler expression for the value function of any reduced game with respect to $\eta$ (see Proposition~\ref{prop:reducedxcore} in Appendix~\ref{sec:appendix_proofs}) i.e.
\begin{equation}\label{eq:reducedxcore}
v_{S, \eta}(T)=\max_{R \subseteq N\setminus S}\Big(v(T \cup R)-\eta(R)\Big)\qquad \forall T\subseteq S.
\end{equation}

The prenucleolus exhibits an interesting \emph{consistency} property for reduced games \cite{peleg2007introduction}.
In fact, the prenucleolus of such a reduced game is simply the projection of the original prenucleolus on the remaining players.
\begin{proposition}\cite{peleg2007introduction}\label{prop:reducednucleolus}
	Let $(N,v)$ be a game, and $\bar{\eta}$ its prenucleolus.
    Then, $\forall S \subsetneq N$, $S\neq \emptyset$
	\begin{equation*}
		\big(\eta(S,v_{S, \bar{\eta}})\big)_j=\bar{\eta}_j \qquad \forall j \in S.
	\end{equation*}
\end{proposition}

We now present some results about reduced games for convex games.
The following proposition is a corollary of Theorem 5.2 in \cite{maschler1971kernel} and Theorem 2.2 in \cite{arin1998characterization}. We give a proof of a slightly more general result in Appendix~\ref{sec:appendix_proofs} (see Proposition~\ref{prop:reducedconvex}).

\begin{proposition}\cite{maschler1971kernel,arin1998characterization}\label{prop:reducedconvex2}
	Let $(N,v)$ be a  convex game and $\eta$ its (pre)nucleolus.
	For all $S \subsetneq N$, $S\neq \emptyset$, the reduced game $(S, v_{S, \eta})$ is convex.
\end{proposition}

As observed earlier, for convex games, the least core value $\optexcess$ and an essential coalition $S$ can be computed in polynomial time via the ellipsoid method and submodular function minimization. For any point $x$ in the least core, we have $x(S) = v(S) + \optexcess$. In particular, if $\eta$ denotes the (pre)nucleolus, then $\eta(S) = v(S) + \optexcess$ is determined. 
Additionally, since $\eta(N) = v(N)$, we also obtain $\eta(N \setminus S)$. 

While the individual components of $\eta$ for players in $S$ and $N \setminus S$ remain unknown, Proposition~\ref{prop:reducedmax2} establishes that no further information is required to define the reduced games with respect to $S$ and $N \setminus S$ for convex games. 
\begin{proposition}\cite{maschler1971kernel}\label{prop:reducedmax2}
	Let $(N,v)$ be a  convex game, $S$ an essential coalition of $(N,v)$, and $\eta$ its (pre)nucleolus. 
	We have:
\begin{equation*}
    \begin{alignedat}{2}
      v_{S, \eta}(T) &= \max \{v(T),\; v(T \cup (N \setminus S))-\eta(N \setminus S)\} 
    && \ \ \ \ \forall T \subseteq S, \\
    v_{N \setminus S, \eta}(T) &= \max \{v(T),\; v(T \cup S)-\eta(S)\} 
    && \ \ \ \ \forall T \subseteq N \setminus S. 
    \end{alignedat}
\end{equation*}
\end{proposition}

This proposition is a simplified form of Lemma 5.7 in \cite{maschler1971kernel}. However, it does not appear explicitly in that form in \cite{maschler1971kernel}, and extracting it is not immediate.  
This result is crucial, as it corrects a fundamental mistake in \cite{kuipersPolynomialTimeAlgorithm1996}, who erroneously claimed that \cite{maschler1971kernel} proves the stronger but Incorrect Theorem~\ref{prop:kuiperswrong} (see Section~\ref{sec:algoRG}). For completeness, we provide a short and direct proof of a more general result in the Appendix~\ref{sec:appendix_proofs} (see Proposition~\ref{prop:reducedmax}). Interestingly, while writing this paper, we discovered that a result akin to Proposition~\ref{prop:reducedmax} had already been identified and utilized in \cite{peleg2007introduction} to demonstrate that the prekernel of a convex game coincides with the (pre)nucleolus (Lemma 5.7.2 in \cite{peleg2007introduction}).

\section{A natural reduced game approach for Nucleolus in convex games}\label{sec:algoRG}

The results from the previous section suggest the following procedure (Algorithm~\ref{alg:pseudokuipers}) for computing  for any $i\in N$ the corresponding component $\eta_i$ of the nucleolus.
This approach was suggested in \cite{kuipersPolynomialTimeAlgorithm1996}, while the key ideas behind this approach can already be traced back to \cite{maschler1971kernel}. In the algorithm, we maintain the invariant that each $(N^{k},v^{k})$ is convex and corresponds to a reduced game induced by the nucleolus. 

\begin{algorithm}[H]
	\caption{Nucleolus through reduced games}\label{alg:pseudokuipers}
	\begin{algorithmic}[1]
    \begin{inputblock}
    A convex game $\left(N,v\right)$ given through an evaluation oracle
    \end{inputblock}
    \begin{outputblock}
    $\eta$, the nucleolus of $(N,v)$
    \end{outputblock}
		\For  {$i \in  N$}\label{line:begin-for-nucl1}
        \State $k\algset 1$
        \State $v^k\algset v$
        \State $N^k\algset N$
		\While{$\abs{N^k} >1$}\label{line:begin-while-nucl1}
		\State Compute an essential coalition ${S}^{k}$ and the least core value $\varepsilon^k$ of game $(N^k,v^k)$ \label{line:essential-coal-nucl1}
        \State\textbf{if } $i\in {S}^{k}$ \textbf{ then } define $N^{k+1}\algset {S}^{k}$ and $\zeta^k \algset v^k(N^k)-(v^k(S^k)+\varepsilon^k)$
        \State \textbf{else } define $N^{k+1}\algset N^k \setminus {S}^{k}$ and $\zeta^k \algset v^k(S^k)+\varepsilon^k$
        \State Define $T^k\algset N^{k}\setminus N^{k+1}$\label{line:define-T_k}
		\State Define $v^{k+1}(T)= \max\{v^{k}(T), \; v^{k}(T \cup T^k) - \zeta^k \} \ \ \forall T\subseteq N^{k+1}$\label{line:define-v_k+1}
		\State $k\algset k+1$
		\EndWhile \label{line:end-while-nucl1}
		\State $\eta_i\algset v^k(N^k)$ 
		\EndFor \label{line:end-for-nucl1}
        \returnstmt{$\eta$}
	\end{algorithmic}
\end{algorithm}

Let us assess the correctness of Algorithm~\ref{alg:pseudokuipers}.
Consider $i\in N$. 
We define $\eta^k \coloneqq \eta(N^k,v^k)$.
Observe that $\zeta^k=\eta^k(T^k)$. 
Since $(N^1, v^1)=(N,v)$ is convex, by Proposition~\ref{prop:reducedmax2} we have $v^{2}=v^{1}_{N^{2},\eta^1}$, that is, $(N^{2},v^{2})$ is a reduced game of $(N^{1},v^{1})$ with respect to its (pre)nucleolus and hence by Proposition~\ref{prop:reducedconvex2} is convex as well.
Using a similar argument, by induction we can prove that, for all $k$, $(N^{k+1},v^{k+1})$ is defined as the reduced game of $(N^{k},v^{k})$ with respect to $N^{k+1}$ and $\eta^k$ , and thus is convex.
Furthermore, by Proposition~\ref{prop:reducednucleolus}, the components of the nucleolus computed in the reduced games are equal to the corresponding components of the original nucleolus $\eta$, that is $\eta^k_j=\eta_j$ for all $j\in N^{k}$. 
Note also that when $\abs{N^k}\le1$ we have indeed $N^k=\{i\}$, as any essential coalition $S^{k-1}$ of $(N^{k-1}, v^{k-1})$ satisfies $S^{k-1} \neq \emptyset, N^{k-1}$ if $k\ge 2$; while if $k=1$ we have $N^1=N=\{i\}$. 
Thus if $\abs{N^k}\le1$ we have $v^k(N^k)=\eta^k_i=\eta_i$.

Observe that step~\ref{line:essential-coal-nucl1} can be done efficiently with the ellipsoid method  \cite{faigle2001computation}, {\em if} it is possible to evaluate $v^k$ efficiently. 
By induction on $k$, we can derive the following formula for $v^{k+1}(T)$, with $[k]=\{1, \ldots, k\}$ for any integer $k\ge 1$ and $[0]=\emptyset$:
\begin{equation}\label{eq:usualtrick}
    v^{k+1}(T)= \max_{J \subseteq [k]} \Bigg(v\Big(T \cup \bigcup_{j\in J} T^j\Big) - \sum_{j\in J} \zeta^j\Bigg) .
\end{equation}
If implemented naively, computing $v^{k+1}(T)$ will involve $O(2^k)$ evaluations of $v$. In order to maintain polynomiality, Theorem 12 in \cite{kuipersPolynomialTimeAlgorithm1996} claims that:

\begin{assertion}\cite{kuipersPolynomialTimeAlgorithm1996}\label{prop:kuiperswrong}
	Let $(N, v)$ be a convex game, let $U \subsetneq N$ be such that $\mbox{y(U)=z(U)}$ $\forall y, z \in \leastcore(v)$, and let $x \in \core(v)$.
	Then for all $S \subsetneq U$, $S \neq  \emptyset$ we have
	\begin{equation*}
		\max_{Q\subseteq N\setminus U} \bigl(v(S\cup Q) - x(Q)\bigr) = \max\{v(S), \; v(S \cup (N\setminus U)) - x(N\setminus U)\}.
	\end{equation*}
\end{assertion} 

Incorrect Theorem~\ref{prop:kuiperswrong} is used in \cite{kuipersPolynomialTimeAlgorithm1996} to evaluate $v^{k+1}(T)$  as 
$\max \{v(T), \; v(T  \cup \bigcup_{j=1}^k T^j) - \sum_{j=1}^{k} \zeta^j\}$, 
for $T\subseteq N^{k+1}$ 
(recalling that $\zeta^j=\eta^j(T^j)=\eta(T^j)$).
Unfortunately, Incorrect Theorem~\ref{prop:kuiperswrong} is not valid (a counterexample is given in Appendix~\ref{sec:appendix_counterexample}). 
We thus start by fixing this issue by providing a procedure to evaluate  $v^{k+1}(T)$ efficiently through Proposition~\ref{prop:reducedmax2}. 
Defining
\begin{equation*}\label{eq:fTT}
    f_{T, T^1,...,T^k}(J) 
    = v(T \cup \bigcup_{j\in J} T^j) - \sum_{j\in J} \zeta^j
    \qquad \forall J \subseteq [k],
\end{equation*}
we have 
\[v^{k+1}(T)=\max_{J \subseteq [k]}  f_{T, T^1,...,T^k}(J).\]
Now observe that, for all $T, T^1,...,T^k$ disjoint, the set function $f_{T, T^1,...,T^k}\colon \powerset{[k]} \rightarrow \mathbb{R}$ is supermodular.
In fact by exploiting the supermodularity of $v$ we have $\forall I, J \subseteq [k]$
\begin{align*}
        &f_{T, T^1,...,T^k}(I)+f_{T, T^1,...,T^k}(J)\\
    =&v\Big(T \cup \bigcup_{j\in I} T^j\Big) - \sum_{j\in I} \zeta^j+v\Big(T \cup \bigcup_{j\in J} T^j\Big) - \sum_{j\in J} \zeta^j\\
    \le& v\Big(T \cup \bigcup_{j\in I \cup J} T^j\Big) +v\Big(T \cup \bigcup_{j\in I \cap J} T^j\Big) - \sum_{j\in I \cup J} \zeta^j - \sum_{j\in I \cap J} \zeta^j\\
    =&f_{T, T^1,...,T^k}(I \cup J)+f_{T, T^1,...,T^k}(I \cap J).
\end{align*}
Therefore $v^{k+1}(T)$ can be computed through submodular function minimization. This proves that it is possible to compute the nucleolus of $(N,v)$ in (oracle) polynomial time through Algorithm~\ref{alg:pseudokuipers}. 

In \cite{kuipersPolynomialTimeAlgorithm1996} a combinatorial approach for step~\ref{line:essential-coal-nucl1} in Algorithm~\ref{alg:pseudokuipers} is proposed.
This approach relies on solving a sequence of feasibility tests for different values of~$\varepsilon$, cubic in~$n$.
Each such test reduces to a non-emptiness test of the base polytope of extended polymatroids associated with crossing submodular functions. 
A combinatorial algorithm for this problem was already available in~\cite{fujishige1984structures}, and a similar approach is described in~\cite{schrijverCombinatorialOptimizationPolyhedra2004}.
In the following section, we propose a simpler and more efficient combinatorial
algorithm for least-core computation. This leads to the following
result.

\begin{theorem}
There exists a combinatorial algorithm to compute the nucleolus of convex games in oracle time 
$O(n^3 \cdot (p_n)^3 \cdot \oracletime[v])$, where $p_n$ denotes the best-known polynomial bound on the number of value-oracle calls required for submodular function minimization, and $\oracletime[v]$ is the time complexity of a single oracle evaluation of the set function~$v$.
\end{theorem}

\noindent We now turn to the main combinatorial ingredient enabling this framework.

\section{The Least Core of Convex Games Combinatorially}\label{sec:leastcore}
In this section we present a new strongly polynomial-time combinatorial algorithm for solving
Problem (\ref{eq:least-core}) and retrieving an essential coalition.   Observe that $\optexcess$ is the maximum $\varepsilon \in \mathbb{R}$ such that the polyhedron defined by the following system is non-empty:
\begin{equation}\label{eq:nonemptiness}
    \begin{alignedat}{2}
x\left(S\right) & \geq v(S) + \excess& \qquad\forall S\subsetneq N,\; S\neq\emptyset \\
x\left(N\right) & =v\left(N\right). 
\end{alignedat}
\end{equation}
We start by showing how to test whether such a polyhedron is non-empty for a given value of $\varepsilon$.
Then we describe how to build upon this procedure to identify $\optexcess$ efficiently (and combinatorially).

\subsection{Non-Emptiness testing\label{subsec:crossing-supermodular-core-non-emptiness}}
Observe first that the function $v_\excess$, defined as 
$v_\excess(S)\coloneqq v(S)+\varepsilon$ for all $S\subsetneq N$, $S\neq\emptyset$, 
with $v_\excess(\emptyset)=0$ and $v_\excess(N)=v(N)$, is crossing supermodular. 
Hence testing the feasibility of (\ref{eq:nonemptiness}) reduces to testing the 
non-emptiness of the core of a game with a crossing supermodular value function.

This problem can be solved in strongly polynomial time using an algorithm by 
Fujishige \cite{fujishige1984structures} for testing the non-emptiness of the base 
polytope of an extended polymatroid associated with a crossing submodular function 
(see also \cite{schrijverCombinatorialOptimizationPolyhedra2004}). 
This is also the approach taken in \cite{kuipersPolynomialTimeAlgorithm1996}. 

Our goal here is to avoid this framework by reformulating the feasibility 
problem as a discrete sandwich problem. This viewpoint leads to a simpler structural 
characterization and ultimately improves the complexity of the least-core computation 
developed later. The key ingredient is Theorem 
\ref{prop:core-non-emptiness-condition}, which builds upon 
Theorem~\ref{thm:Franks-discrete-sandwich-theorem} (Frank’s discrete sandwich theorem 
\cite{frankAlgorithmSubmodularFunctions1982a}) and provides a more direct 
non-emptiness test.

\begin{theorem}\cite{frankAlgorithmSubmodularFunctions1982a} 
\label{thm:Franks-discrete-sandwich-theorem}
Let $V$ be a finite set, let $f:2^V\rightarrow \mathbb R$ be a supermodular function, and let $g:2^V\rightarrow \mathbb R$ be a submodular function with $f(\emptyset) \le 0 \le g(\emptyset)$.
Then there exists a vector $x\in\mathbb{R}^{|V|}$ satisfying
\[
f\left(U\right)\leq x\left(U\right)\leq g\left(U\right)\qquad\forall U\subseteq V,
\]
if and only if $f\left(U\right)\leq g\left(U\right)$ for all $U\subseteq V$ (or, equivalently, if and only if $\min_{U\subseteq V}\left(g\left(U\right)-f\left(U\right)\right)\geq0$).
\end{theorem}

Observe now that for a game $(N,v)$, the inequality $x(S)= \sum_{P\in {\mathcal P}} x(P)\geq  \sum_{P\in {\mathcal P}} v(P)$ is valid for $\core\left(v\right)$, for any $\emptyset\neq S\subseteq N$ and any (proper) partition $\mathcal P$ of $S$. 
Hence we can equivalently replace $v$ by $\supaddcover v$ in the inequalities defining $\core\left(v\right)$ 
where the set function $\supaddcover v$, defined by
\[\supaddcover v\left(S\right)\coloneqq\max \big\{ \sum\nolimits_{P\in\mathcal{P}}v\left(P\right)\; \big|\;\mathcal{P}\subseteq\powerset S \allowbreak \text{proper partition of \ensuremath{S}}\big\}  \quad
\forall  \emptyset \neq S\subseteq N \]
and $\supaddcover v (\emptyset)=0$, 
is called the {\em superadditive cover} of $v$.
It is easy to prove that the superadditive cover is indeed superadditive.
An equivalent concept was used in \cite{fujishige1984structures} and \cite{schrijverCombinatorialOptimizationPolyhedra2004} to test the non-emptiness of the base polytope of extended polymatroids associated with intersecting and crossing submodular functions.  
Of course for a convex game $(N,v)$ we have $\supaddcover v = v$, due to the superadditivity of $v$. 
It is known that the use of superadditive cover allows to retrieve supermodularity when $v$ is intersecting supermodular. Besides, $\supaddcover v(S)$ can be computed in polytime by a greedy algorithm using submodular function minimization as a subroutine. These properties were proven by \cite{fujishige1984structures} and \cite{schrijverCombinatorialOptimizationPolyhedra2004} in the submodular setting. 
We give a direct proof in Appendix~\ref{sec:appendix_proofs} for sake of completeness (see Proposition~\ref{prop:convexity-superadditive-cover-near-convex-games}). 

\begin{proposition}
\label{prop:convexity-superadditive-cover-near-convex-games2}\cite{fujishige1984structures}\cite{schrijverCombinatorialOptimizationPolyhedra2004}  Let
$v\colon\powerset V\rightarrow\mathbb{R}$ be an intersecting supermodular function, then its superadditive cover $\supaddcover v$ is supermodular on $V$. 
Moreover, $\supaddcover v(S)$ and a  proper partition $\mathcal T$ of $S$ such that $\supaddcover v\left(S\right)=\sum\nolimits_{T\in\mathcal{T}}v\left(T\right)$ can be computed in (oracle) strongly polynomial time for all $S\subseteq V$.
\end{proposition}

We can now state the main result of this section, which shows that testing the non-emptiness of~(\ref{eq:nonemptiness}) for a given $\varepsilon$, which is equivalent to checking whether $\core\left(\excessgame\right)$ is non-empty, can be done in polytime via an approach that is both different from and faster than Fujishige's \cite{fujishige1984structures,kuipersPolynomialTimeAlgorithm1996,schrijverCombinatorialOptimizationPolyhedra2004,frank2011connections}. We analyze the computational complexity in Section~\ref{sec:complexity}. The corresponding result can be extended to testing non emptiness of base polytope associated with crossing submodular functions on families closed under complement (see Theorem~\ref{thm:sandwich-non-emptiness-condition}).

\global\long\def\Ns{N_{s}}%

\begin{theorem}\label{thm:empty}
\label{prop:core-non-emptiness-condition} Let $(N,w)$ be a game with $w\colon\powerset N\rightarrow\mathbb{R}$
 crossing supermodular. Then for any $s\in N$,
it holds that
\[
\core\left(w\right)\neq\emptyset\qquad\text{if and only if}\qquad\min_{S\subseteq N\setminus\left\{ s\right\} }\bigl(-\supaddcover w\left(S\right)-\supaddcover u\left(S\right)\bigr)\geq0,
\]
with $u\left(S\right)\coloneqq w\left(N\setminus S\right)-w\left(N\right)$ for all $S \subseteq N \setminus \{s\}$
and $\supaddcover w$, $\supaddcover u$ superadditive covers of $w$, $u$ respectively. 
Moreover, $-\supaddcover w(S)-\supaddcover u(S)$ is submodular and thus $\min_{S\subseteq N\setminus\{ s\} }\left(-\supaddcover w\left(S\right)-\supaddcover u\left(S\right)\right)$ can be computed in (oracle) strongly polynomial time through submodular function minimization.\footnote{The game $(N,-u)$ is also referred to as the \emph{dual game} of $(N,w)$ in \cite{peleg2007introduction}.
}
\end{theorem}
\proof
For simplicity, we denote $N\setminus\left\{ s\right\} $ as $\Ns$.
By explicitly expressing the component $x_{s}$,
the system of inequalities defining $\core\left(w\right)$ can be
reformulated as follows.
\[
\begin{cases}
\phantom{x\left(N\right)}\mathllap{x\left(S\right)}\geq w\left(S\right) & \forall S\subseteq N,\\
x\left(N\right)=w\left(N\right)
\end{cases}\quad\iff\quad\begin{cases}
\phantom{x\left(\Ns\right)+x_{s}}\mathllap{x\left(S\right)}\geq w\left(S\right) & \forall S\subseteq\Ns,\\
\phantom{x\left(\Ns\right)+x_{s}}\mathllap{x\left(S\right)+x_{s}}\geq w\left(S\cup\left\{ s\right\} \right) & \forall S\subseteq\Ns,\\
x\left(\Ns\right)+x_{s}=w\left(N\right).
\end{cases}
\]
By subtracting the second constraint into the last, we obtain:
\[
\left(x\left(\Ns\right)+\bcancel{x_{s}}\right)-\left(x\left(S\right)+\bcancel{x_{s}}\right)=x\left(\Ns\setminus S\right)\leq w\left(N\right)-w\left(S\cup\left\{ s\right\} \right)\qquad\forall S\subseteq\Ns,
\]
and by taking $S'=\Ns\setminus S$ and consequently $S\cup\left\{ s\right\} =\left(\Ns\setminus S'\right)\cup\left\{ s\right\} =N\setminus S'$,
such inequalities become
\[
x\left(S'\right)\leq w\left(N\right)-w\left(N\setminus S'\right)\qquad\forall S'\subseteq\Ns.
\]
The original system can thus be rewritten as 
\[
\begin{cases}
\phantom{x\left(\Ns\right)}\mathllap{x\left(S\right)}\geq w\left(S\right) & \forall S\subseteq\Ns,\\
\phantom{x\left(\Ns\right)}\mathllap{x\left(S\right)}\leq w\left(N\right)-w\left(N\setminus S\right) \quad & \forall S\subseteq\Ns,\\
x\left(\Ns\right)+x_{s}=w\left(N\right).
\end{cases}
\]
Since no inequalities of the system involve $x_{s}$,
it is uniquely determined by $x_{s}=w\left(N\right)-x\left(\Ns\right)$
for any instantiation of all other components of $x$. The feasibility
of the system (and non-emptiness of $\core\left(w\right)$) is therefore
equivalent to the existence of a vector $x$ in the subspace $\mathbb{R}^{|\Ns|}$
of $\mathbb{R}^{n}$ for which
\begin{equation}
w\left(S\right)\leq x\left(S\right)\leq w\left(N\right)-w\left(N\setminus S\right)=-u(S)\qquad\forall S\subseteq\Ns.\label{eq:non-emptiness-condition-crossing-supermodular}
\end{equation}
Any $x$ satisfying (\ref{eq:non-emptiness-condition-crossing-supermodular})
must also satisfy each of the following 
\begin{align}
\max_{\text{\ensuremath{\mathcal{P}} prop.\ part.\,of \ensuremath{S}}}\sum\nolimits_{P\in\mathcal{P}}w\left(P\right) & \leq x\left(S\right)\leq\min_{\text{\ensuremath{\mathcal{Q}} prop.\ part.\,of \ensuremath{S}}}\sum\nolimits_{Q\in\mathcal{Q}}-u(Q)& \forall S\subseteq\Ns,\nonumber \\
\supaddcover w(S) & \leq x(S) \leq - \supaddcover u(S)& \forall S\subseteq\Ns,\label{eq:non-emptiness-conditionich-superadditive-cover}
\end{align}
and vice versa.  We claim now that both $u$ and $w$ are crossing supermodular. This is by hypothesis for $w$. Let us prove it for $u$. Let $S,T\subseteq N$ such that $S\cap T\neq\emptyset$ and $S\cup T\neq N$,
then we have:
\begin{equation*}
    \begin{alignedat}{1}
   &u\left(S\cap T\right)+u\left(S\cup T\right) \\
    =&\Big(w\big(N\setminus\left(S\cap T\right)\big)-w\left(N\right)\Big)+\Big(w\big(N\setminus\left(S\cup T\right)\big)-w\left(N\right)\Big)\\
  =&w\big(\left(N\setminus S\right)\cup\left(N\setminus T\right)\big)+w\big(\left(N\setminus S\right)\cap\left(N\setminus T\right)\big)-2\cdot w\left(N\right)\\
  \geq&\Big(w\left(N\setminus S\right)-w\left(N\right)\Big)+\Big(w\left(N\setminus T\right)-w\left(N\right)\Big)=u\left(S\right)+u\left(T\right).     
    \end{alignedat}
\end{equation*}

Since both $u$ and $w$ are crossing supermodular,
they are
 intersecting supermodular {\em on $N_s$}.  Then, by Proposition~\ref{prop:convexity-superadditive-cover-near-convex-games2},  $\supaddcover w(S)$ and $\supaddcover u(S)$ are supermodular on  $N_s$ and so  $-\supaddcover w(S)-\supaddcover u(S)$ is submodular on $N_s$. Thanks to Frank's discrete sandwich theorem  (Theorem~\ref{thm:Franks-discrete-sandwich-theorem}),
the existence of an $x$ satisfying (\ref{eq:non-emptiness-conditionich-superadditive-cover})
is equivalent to the condition
\[
\min_{S\subseteq\Ns}\bigl(-\supaddcover w\left(S\right)-\supaddcover u\left(S\right)\bigr)\geq0,
\]
which can be computed in polytime using submodular function minimization.
\endproof

\subsection{A simple Farkas' Lemma-driven procedure to search for  \texorpdfstring{$\optexcess$}{TEXT}}\label{sec:farkas}

Farkas' Lemma \cite{farkas1902theorie} states that exactly one of the two systems below have a solution for any given~$\excess$.
\begin{center}
\begin{minipage}[c]{0.45\columnwidth}%
\begin{align}
x(S) & \geq v(S)+\excess\quad\forall S\subsetneq N,S\neq\emptyset\nonumber \\
 & \tag{P}\label{eq:primal-system}\\
x(N) & =v(N)\nonumber 
\end{align}
\end{minipage}\hspace{0.05\columnwidth}%
\begin{minipage}[c]{0.50\columnwidth}%
\begin{align}
\sum_{\hspace*{1.25em}\mathclap{S\subseteq N:s\in S}\hspace*{1.25em}}\lambda_{S} & =0\quad\forall s\in N\nonumber \\
\sum_{\hspace*{1.25em}\mathclap{S\subsetneq N, S \neq \emptyset}\hspace*{1.25em}}\lambda_{S} & \cdot(v(S)+\excess)+\lambda_{N}\cdot v(N)>0\tag{D}\label{eq:dual-system}\\
\lambda_{S} & \geq0\quad\forall S\subsetneq N,S\neq\emptyset\nonumber. 
\end{align}
\end{minipage}
\end{center}
Now observe that when $\excess$ yields to infeasibility of the first system (\ref{eq:primal-system}), Theorem~\ref{thm:empty} provides a solution (Farkas certificate) to (\ref{eq:dual-system}). 
Indeed, defining $u_\excess\left(S\right)\coloneqq v_\excess\left(N\setminus S\right)-v_\excess\left(N\right)$ for all $S\subseteq N\setminus \{s\}$ for some $s\in N$, 
we have $\supaddcover{\excessgame[\excess]}\left(\overline{S}\right)+\supaddcover{u_{\excess}}\left(\overline{S}\right)>0$ for some $\overline{S}\subseteq N \setminus \{s\}$, $\overline{S} \neq \emptyset$
where $\supaddcover{u_{\excess}}$, $\supaddcover{v_{\excess}}$ are the superadditive covers of $u_{\excess}$, $v_{\excess}$ respectively (that can be computed in polytime due to Proposition~\ref{prop:convexity-superadditive-cover-near-convex-games2}).  
This is equivalent to 
\begin{equation}
-\left|\mathcal{Q}\right|v\left(N\right)+\sum_{P\in\mathcal{P}}v\left(P\right)+\sum_{Q\in\mathcal{Q}}v\left(N\setminus Q\right)+\left(\left|\mathcal{P}\right|+\left|\mathcal{Q}\right|\right)\cdot\excess>0,\label{eq:continuing-condition}
\end{equation}
for $\mathcal P$ and $\mathcal Q$ the proper partitions of $\overline{S}$ such that $\supaddcover{v_{\excess}}\left(\overline{S}\right) = \sum_{P\in {\mathcal P}} (v(P)+ \excess)$ and  
$\supaddcover{u_{\excess}}\left(\overline{S}\right) = {\sum_{Q\in {\mathcal Q}} (v(N\setminus Q) - v(N) + \excess)} $ (that again can be computed in polytime due to Proposition~\ref{prop:convexity-superadditive-cover-near-convex-games2}).  
A Farkas certificate is thus (as each player is covered $|{\mathcal Q}|$ times by the elements of the family \mbox{$\mathcal S\coloneqq\mathcal P\cup \{N\setminus Q \mid Q\in \mathcal Q\}$)}
\[
\lambda_{S}=\begin{cases}
1 & \text{if \ensuremath{S\in\mathcal{P}} or \ensuremath{N\setminus S\in\mathcal{Q}}},\\
-|{\mathcal Q}| \quad& \text{if \ensuremath{S=N}}\\
0 & \text{otherwise},
\end{cases}
\]
which also shows that
\[
\optexcess\leq\frac{1}{\left|\mathcal{P}\right|+\left|\mathcal{Q}\right|}\Bigl(\left|\mathcal{Q}\right|v\left(N\right)-\sum_{P\in\mathcal{P}}v\left(P\right)-\sum_{Q\in\mathcal{Q}}v\left(N\setminus Q\right)\Bigr).
\]

It is natural to use this upper bound as the next candidate for $\varepsilon$.
This motivates Algorithm~\ref{alg:least-core}, which performs a monotone
one-dimensional search on the least core value $\varepsilon$, where infeasibility
certificates derived from Farkas’ lemma are used to update the candidate bound.
Algorithm~\ref{alg:least-core} can be viewed as a Discrete Newton Method~\cite{goemans2017discrete}; see Appendix~\ref{sec:appendix_newton} for a detailed description from this perspective.
We refrain from adopting this formalization in the main text to keep the presentation self-contained and accessible.

\begin{algorithm}
\begin{algorithmic}[1]
\begin{inputblock}
A convex game $\left(N,v\right)$ given by an evaluation oracle and a bound $M > \max_{T\subseteq N} \abs{v(T)}$

\end{inputblock}

\begin{outputblock}
The least core value $\optexcess$ and a non-empty set of essential coalitions $\mathcal S$

\end{outputblock}

\State{Select $s\in N$ and initialize $\excess'\algset2M$}

\Repeat
\State{$\excess \algset\excess'$}
\State{Define $v_\excess(S)\coloneqq v(S)+\varepsilon$ for all $S\subsetneq N$, $S \neq \emptyset$, $v_\excess(\emptyset)=0$, and $v_\excess(N)=v(N)$}
\State{Define $u_{\excess}\left(S\right)\coloneqq\excessgame\left(N\setminus S\right)-\excessgame\left(N\right)$ for all $S\subseteq N$}
\State{$\overline{S}\algset\vphantom{\bigg|}\argmin_{S\subseteq N\setminus\left\{ s\right\} }\left(-\supaddcover{\excessgame}\left(S\right)-\supaddcover{u_{\excess}}\left(S\right)\right)$, with $\overline{S}\neq\emptyset$ if possible}
\State{Let $\mathcal{P}$ prop.\ partition of $\overline{S}$ s.t.
$
\supaddcover{\excessgame}(\overline{S})  =\sum_{P\in\mathcal{P}}\excessgame(P)=\sum_{P\in\mathcal{P}}v(P)+|\mathcal{P}|\excess$\label{line:least-core-partitions1}}
\State{Let $\mathcal{Q}$ prop.\ partition of $\overline{S}$ s.t.
$\supaddcover{u_{\excess}}(\overline{S})  =\sum_{Q\in\mathcal{Q}}u_{\excess}(Q)=\sum_{Q\in\mathcal{Q}}v(N\setminus Q)+|\mathcal{Q}|\excess-|\mathcal{Q}|v(N)$
\label{line:least-core-partitions2}}
\State{$\excess'\algset\frac{1}{\left|\mathcal{P}\right|+\left|\mathcal{Q}\right|}\Bigl(\left|\mathcal{Q}\right|v\left(N\right)-\sum_{P\in\mathcal{P}}v\left(P\right)-\sum_{Q\in\mathcal{Q}}v\left(N\setminus Q\right)\Bigr)$\label{line:least-core-update}}
\Until{$-\supaddcover{\excessgame}(\overline{S})-\supaddcover{u_{\excess}}(\overline{S})\geq0$}

\returnstmt{$\excess$, $\mathcal S\coloneqq\mathcal P\cup \{N\setminus Q \mid Q\in \mathcal Q\}$}

\end{algorithmic}

\caption{\label{alg:least-core}Least core}
\end{algorithm}

Note that one needs to identify first a value $\excess$ that is unfeasible. 
We can compute such a value from a bound $M$ on $\max_{T\subseteq N}\left|v\left(T\right)\right|$. Such a bound can be computed by taking the maximum bound on $A=\max_{T\subseteq N} v\left(T\right)$ and $B=\max_{T\subseteq N} -v\left(T\right)$: $A$  can be computed in polytime through submodular function minimization, while a bound on B can be obtained from an approximation algorithm for submodular function maximization.  
Indeed consider the non-negative submodular function $w\coloneqq-v + A$.   The greedy algorithm is a 1/3-approximation algorithm for the problem $\max_{T\subseteq N} w(T)$ \cite{buchbinder2015tight}.
The value  $C$ returned by the greedy algorithm on $w$ thus satisfies  $C\geq \frac{1}{3} \max_{T\subseteq N} (-v(T)+A)$.
Hence $3C - A$ is an upper bound on $B$.
$M=\max\{A,3C-A\}+1$ is thus a bound on $\max_{T\subseteq N} \abs{v(T)}$ that can be computed in polytime. 

\begin{proposition}
\label{prop:least-core-computation2}Algorithm~\ref{alg:least-core} terminates
in at most $2n-2$ iterations with the least core value
$\optexcess$ of the convex
game $\left(N,v\right)$ and a (non-empty) set of essential coalitions $\mathcal S$.
\end{proposition}

\proof
We denote by $\excess^{i}$ the value of $\excess$ at iteration $i$ of the repeat-until loop,  by $\overline{S}_i$ the minimizer $\overline{S}$, and by $\mathcal{P}^{i}$, $\mathcal{Q}^{i}$ the associated partitions determined in 
lines~\ref{line:least-core-partitions1}-\ref{line:least-core-partitions2}, which can be computed through Algorithm~\ref{alg:superadditive-cover}. Notice that when the termination condition is not met, lines~\ref{line:least-core-partitions1}-\ref{line:least-core-partitions2} are well-defined as $\bar{S}_i\neq \emptyset$. We discuss the case where the termination criterion is met later. 

We have $\excess^{1}\coloneqq2M$ and, at each iteration $i$ where the terminating condition is not met (a non-terminating iteration $i$ hereafter):

\begin{equation*}
    \begin{alignedat}{1}
\supaddcover{\excessgame[\excess^{i}]}(\overline{S}_i)=\sum_{P\in\mathcal{P}^{i}}\excessgame[\excess^{i}]\left(P\right),\qquad\supaddcover{u_{\excess^{i}}}(\overline{S}_i)=\sum_{Q\in\mathcal{Q}^{i}}u_{\excess^{i}}\left(Q\right),  \\
\excess^{i+1}\coloneqq\frac{1}{\left|\mathcal{P}^{i}\right|+\left|\mathcal{Q}^{i}\right|}\Big(\left|\mathcal{Q}^{i}\right|v\left(N\right)-\sum_{P\in\mathcal{P}^{i}}v\left(P\right)-\sum_{Q\in\mathcal{Q}^{i}}v\left(N\setminus Q\right)\Big).
    \end{alignedat}
\end{equation*}
Remember how, unraveling the definitions of $\supaddcover{v_{\excess_i}}$ and $\supaddcover{u_{\excess_i}}$, we can rewrite similarly to (\ref{eq:continuing-condition}) the LHS of the terminating condition as
\begin{equation}
-\supaddcover{v_{\excess_i}}(\overline{S}_i) - \supaddcover{u_{\excess_i}}(\overline{S}_i)=\left|\mathcal{Q}^{i}\right|v(N)-\sum_{P\in\mathcal{P}^{i}}v(P)-\sum_{Q\in\mathcal{Q}^{i}}v(N\setminus Q)-\left(\left|\mathcal{P}^i\right|+\left|\mathcal{Q}^i\right|\right)\excess^{i}. \label{eq:explicit-continuing-condition}
\end{equation}

We prove first that the initial $\excess = \excess^1$
never satisfies the terminating condition by proving that the inequality
(\ref{eq:continuing-condition}) holds for any proper partitions $\mathcal{P},\mathcal{Q}$ of any subset $S\subseteq N$, $S \neq \emptyset$, as  
indeed
it holds
\begin{equation*}
    \begin{alignedat}{1}
  &-|\mathcal{Q}|v(N)+\sum_{P\in\mathcal{P}}v(P)+\sum_{Q\in\mathcal{Q}}v(N\setminus Q)+(|\mathcal{P}|+|\mathcal{Q}|)\excess^{1} \\
   >&-|\mathcal{Q}| M-|\mathcal{P}| M-|\mathcal{Q}| M+(|\mathcal{P}|+|\mathcal{Q}|)2M \\
  \ge& |\mathcal{P}| M>0.   
    \end{alignedat}
\end{equation*}
Now we prove that, for a non-terminating iteration $i$, we have  $\excess^{i+1}<\varepsilon^i$. It implies in particular: 
\begin{equation}
   \excess^{j+1} <\varepsilon^{i+1}  \mbox{ for all non-terminating iterations  $i,j$ with } i< j. \label{eq:count}
\end{equation}
Assume that at iteration $i$ the procedure does not
terminate, then (\ref{eq:continuing-condition}) holds with $\excess=\excess^{i}$,
$\mathcal{P}=\mathcal{P}^{i}$ and $\mathcal{Q}=\mathcal{Q}^{i}$. Equation
(\ref{eq:continuing-condition}) can be rewritten in the form
\[
\frac{1}{\left|\mathcal{P}^{i}\right|+\left|\mathcal{Q}^{i}\right|}\Big(\left|\mathcal{Q}^{i}\right|v\left(N\right)-\sum_{P\in\mathcal{P}^{i}}v\left(P\right)-\sum_{Q\in\mathcal{Q}^{i}}v\left(N\setminus Q\right)\Big)<\excess^{i},
\]
but since the LHS is the definition of $\excess^{i+1}$ (as $i$ is a non-terminating iteration), we obtain
$\excess^{i+1}<\excess^{i}$. 

From the arguments of Section~\ref{sec:farkas} we already know that at any iteration $i$ the value of $\excess^{i}$ is an upper bound on the least core value $\optexcess$, and that if the algorithm terminates it thus does so with the correct value.

We now deal with the termination of the algorithm. Let us prove that for any non-terminating iterations $j\neq i$, it must be that $\left|\mathcal{P}^{i}\right|+\left|\mathcal{Q}^{i}\right|\neq\left|\mathcal{P}^{j}\right|+\left|\mathcal{Q}^{j}\right|$.
Suppose by contradiction  that instead $\left|\mathcal{P}^{i}\right|+\left|\mathcal{Q}^{i}\right|=\left|\mathcal{P}^{j}\right|+\left|\mathcal{Q}^{j}\right|$. 
We claim  that
\begin{equation}
\left|\mathcal{Q}^{i}\right|v\left(N\right)-\sum_{P\in\mathcal{P}^{i}}v\left(P\right)-\sum_{Q\in\mathcal{Q}^{i}}v\left(N\setminus Q\right)=\left|\mathcal{Q}^{j}\right|v\left(N\right)-\sum_{P\in\mathcal{P}^{j}}v\left(P\right)-\sum_{Q\in\mathcal{Q}^{j}}v\left(N\setminus Q\right).\label{eq:equality-partition-expression}
\end{equation}
Indeed assume by contradiction that \textsc{w.l.o.g.}~we have 
\[
\left|\mathcal{Q}^{i}\right|v(N)-\sum_{P\in\mathcal{P}^{i}}v(P)-\sum_{Q\in\mathcal{Q}^{i}}v(N\setminus Q)>\left|\mathcal{Q}^{j}\right|v(N)-\sum_{P\in\mathcal{P}^{j}}v(P)-\sum_{Q\in\mathcal{Q}^{j}}v(N\setminus Q).
\]
It follows that by (\ref{eq:explicit-continuing-condition})
\begin{equation*}
    \begin{alignedat}{1}
    -\supaddcover{v_{\excess_i}}(\overline{S}_i) - \supaddcover{u_{\excess_i}}(\overline{S}_i) 
&=\left|\mathcal{Q}^{i}\right|v(N)-\sum_{P\in\mathcal{P}^{i}}v(P)-\sum_{Q\in\mathcal{Q}^{i}}v(N\setminus Q)-\left(\left|\mathcal{P}^i\right|+\left|\mathcal{Q}^i\right|\right)\excess^{i}  \\
&>\left|\mathcal{Q}^{j}\right|v(N)-\sum_{P\in\mathcal{P}^{j}}v(P)-\sum_{Q\in\mathcal{Q}^{j}}v(N\setminus Q)-\left(\left|\mathcal{P}^j\right|+\left|\mathcal{Q}^j\right|\right)\excess^{i}\\
&=\sum_{P \in \mathcal{P}^j} - v_{\excess_i}(P) + \sum_{Q \in \mathcal{Q}^j} - u_{\excess_i}(Q),
    \end{alignedat}
\end{equation*}
but then 
\begin{equation*}
    \begin{alignedat}{1}
    -\supaddcover{v_{\excess_i}}(\overline{S}_j) - \supaddcover{u_{\excess_i}}(\overline{S}_j) 
    &= \min_{\mathcal{P} \text{ prop.\ part.\ of } \overline{S}_j} \Bigl( \sum_{P \in \mathcal{P}} - v_{\excess_i}(P) \Bigr)
    +\min_{\mathcal{Q} \text{ prop.\ part.\ of } \overline{S}_j} \Bigl( \sum_{Q \in \mathcal{Q}} - u_{\excess_i}(Q) \Bigr)\\
    &\leq \sum_{P \in \mathcal{P}^j} - v_{\excess_i}(P) + \sum_{Q \in \mathcal{Q}^j} - u_{\excess_i}(Q) \\
    &< -\supaddcover{v_{\excess_i}}(\overline{S}_i) - \supaddcover{u_{\excess_i}}(\overline{S}_i),
    \end{alignedat}
\end{equation*}
contradicting that  $\overline{S}_i$ is minimizing $\left(-\supaddcover{v_{\excess_i}}\left(S\right)-\supaddcover{u_{\excess_i}}\left(S\right)\right)$ over all $S\subseteq N\setminus\left\{ s\right\}$. 

Now the validity of (\ref{eq:equality-partition-expression}) in turn implies that 
\begin{equation*}
    \begin{alignedat}{1}
   \excess^{i+1} & =\frac{\left|\mathcal{Q}^{i}\right|v(N)-\sum_{P\in\mathcal{P}^{i}}v(P)-\sum_{Q\in\mathcal{Q}^{i}}v(N\setminus Q)}{\left|\mathcal{P}^{i}\right|+\left|\mathcal{Q}^{i}\right|}\\
& = \frac{\left|\mathcal{Q}^{j}\right|v(N)-\sum_{P\in\mathcal{P}^{j}}v(P)-\sum_{Q\in\mathcal{Q}^{j}}v(N\setminus Q)}{\left|\mathcal{P}^{j}\right|+\left|\mathcal{Q}^{j}\right|}=\excess^{j+1}, 
    \end{alignedat}
\end{equation*}
contradicting (\ref{eq:count}). Since $\left|\mathcal{P}^{i}\right|,\left|\mathcal{Q}^{i}\right|\in\left\{ 1,\ldots,n-1\right\}$ we have $2\leq\left|\mathcal{P}^{i}\right|+\left|\mathcal{Q}^{i}\right|\leq2n-2$. The number of non-terminating iterations is thus at most $2n-3$.

Notice that, when the procedure
terminates (say at iteration $k\geq 2$), the terminating condition is met with
equality by $\bar{S}_{k-1}\neq \emptyset$, since it must be that $\min_{S\subseteq N\setminus\left\{ s\right\} }\left(-\supaddcover{\excessgame[\excess^{k}]}\left(S\right)-\supaddcover{u_{\excess^{k}}}\left(S\right)\right)\geq0$
and, as $\mathcal{P}^{k-1}$ and $\mathcal{Q}^{k-1}$ are proper partitions of $\bar{S}_{k-1}$, by the definition of $\excess^{k}$
\begin{equation*}
    \begin{alignedat}{1}
    &-\supaddcover{\excessgame[\excess^{k}]}(\bar{S}_{k-1})-\supaddcover{u_{\excess^{k}}}(\bar{S}_{k-1})
    \leq&\left|\mathcal{Q}^{k-1}\right|v(N)-\sum_{P\in\mathcal{P}^{k\mathrlap{-1}}}v(P)-\sum_{Q\in\mathcal{Q}^{k\mathrlap{-1}}}v(N\setminus Q)-\left(\left|\mathcal{P}^{k-1}\right|+\left|\mathcal{Q}^{k-1}\right|\right)\excess^{k}
=0.
    \end{alignedat}
\end{equation*}

The minimum at iteration $k$ is thus attained by the non-empty set $\bar{S}_{k-1}$, so that $\bar{S}_{k}\neq \emptyset$ by the choice of $\overline{S}$ (and lines~\ref{line:least-core-partitions1}-\ref{line:least-core-partitions2} are again well-defined). In the terminating iteration $k$, we thus have
\[
\frac{1}{\left|\mathcal{P}^k\right|+\left|\mathcal{Q}^k\right|}\Big(\left|\mathcal{Q}^k\right|v\left(N\right)-\sum_{P\in\mathcal{P}^k}v\left(P\right)-\sum_{Q\in\mathcal{Q}^k}v\left(N\setminus Q\right)\Big)=\optexcess,
\]
The solution $(\mu_S)_{S\subseteq N, S \neq \emptyset}$ defined as
\[
\mu_{S}=\begin{cases}
-\frac{1}{|{\mathcal P}^k| + |{\mathcal Q}^k|}\qquad & \text{if \ensuremath{S\in\mathcal{P}^k} or \ensuremath{N\setminus S\in\mathcal{Q}^k}},\\
\frac{|{\mathcal Q}^k|}{|{\mathcal P}^k| + |{\mathcal Q}^k|} & \text{if \ensuremath{S=N}}\\
0 & \text{otherwise},
\end{cases}
\]
is thus an optimal solution to (\ref{eq:least-core-dual}), as it is a feasible solution (obviously $\mu_S\leq 0$ for $S\neq N$ and $\sum_{S\neq N} \mu_S=-1$, while we have $\sum_{S\neq N:s\in S} \mu_S + \mu_N=0$ as the elements of $N$ are covered exactly $|{\mathcal Q}|$ times by the elements of the family $\mathcal S\coloneqq\mathcal P\cup \{N\setminus Q \mid Q\in \mathcal Q\}$) with optimal value $\optexcess$. Hence, the elements in $\mathcal S$, having negative dual multipliers, are essential coalitions by complementary slackness.
\endproof

\section{Complexity}\label{sec:complexity}

We now briefly discuss the computational complexity of the various algorithms.  Let
$\oracletime[f]$ be the time-complexity of an oracle evaluating a
set function $f\colon\powerset N\rightarrow\mathbb{R}$ over any subset
of $N$ with $n\coloneqq\left|N\right|$.
When $f$
is submodular, we assume that computing a minimizer $S\subseteq N$ of $f$ can be done in time
$\bigo(p_n\cdot\oracletime[f])$, where $p_n$ is a function of $n$ (asymptotically) bounded by a polynomial in $n$. The survey \cite{leeFasterCuttingPlane2015} states that the best-known combinatorial 
algorithm to find a minimizer of $f$ is by Orlin \cite{orlin2009faster} and has $p_n=O(n^{5})$ (if $\oracletime[f]=\Omega(n)$). 

Our new Algorithm~\ref{alg:least-core} for computing the least core, which leverages Frank's discrete sandwich theorem~\cite{frankAlgorithmSubmodularFunctions1982a}, runs in time 
$O(n^2 \cdot (p_n)^2 \cdot \oracletime[v])$. 
In contrast, the alternative algorithm for computing the least core in \cite{kuipersPolynomialTimeAlgorithm1996}, building on the work of Fujishige~\cite{fujishige1984structures}, has a complexity of 
$O(n^5 \cdot (p_n)^2 \cdot \oracletime[v])$. 
Overall, this results in a factor $n^3$ improvement over \cite{kuipersPolynomialTimeAlgorithm1996}. More precisely, while the gain is limited to a factor of $n$ for non-emptiness testing, the certificate of emptiness provided by Theorem~\ref{thm:empty} yields an additional factor $n^{2}$ reduction in the number of iterations required to determine~$\varepsilon^*$. 
As a consequence, Algorithm~\ref{alg:pseudokuipers} for computing the nucleolus can be executed in time 
$O(n^4 \cdot (p_n)^3 \cdot \oracletime[v])$ when using Algorithm~\ref{alg:least-core}. Moreover, by avoiding a separate execution for each $i \in N$ and instead applying a divide-and-conquer strategy, we can save a further factor $n$, yielding an overall complexity of 
$O(n^3 \cdot (p_n)^3 \cdot \oracletime[v])$ for computing the nucleolus.

 A preliminary analysis suggests that the only known algorithm for computing the nucleolus of convex games~\cite{faigle2001computation} runs in time
$ O(n^{14} \cdot (n + \langle v \rangle)^2 \cdot p_n \cdot \oracletime[v])$
where $\langle v \rangle$ denotes an upper bound on the encoding length of the values $v(S)$ for all $S \subseteq N$ (disregarding the elementary operations). This estimate relies on the complexity of optimization via separation, as described in Theorem~4.21 of~\cite{korte2011combinatorial}, and on the fact that a polyhedron in $\mathbb{R}^n$ with facet complexity $\varphi$ has vertex complexity at most $4n^2 \varphi$ (Lemma~6.2.4 in~\cite{grotschel2012geometric}). In the algorithm of \cite{faigle2001computation}, there are $O(n^2)$ iterations, each involving the resolution of a linear program with facet complexity $O(n + \langle v \rangle)$, solved through a separation oracle involving $O(n^2)$ calls to submodular function minimization. 
These preliminary investigations thus suggest that the oracle complexity of our algorithm improves upon~\cite{faigle2001computation}. Nonetheless, given that both complexity analyses are likely not tight, it remains difficult to draw a definitive conclusion.

\section*{Acknowledgments}

We would like to express our sincere gratitude to Michel Grabisch for providing us with a scanned version of Kuipers' original manuscript \cite{kuipersPolynomialTimeAlgorithm1996}, and to Ulrich Faigle for his invaluable assistance in helping us access this important resource. Their support was crucial to the progress of this work.

\section*{Declarations}
No funds, grants, or other support was received.
The authors have no competing interests to declare that are relevant to the content of this article.

\newpage
\renewcommand{\thesection}{A\arabic{section}}
\setcounter{section}{0}

\section{Counterexample to Incorrect Theorem~\ref{prop:kuiperswrong}}\label{sec:appendix_counterexample}
Let us first remind the statement of the Incorrect Theorem~\ref{prop:kuiperswrong}.

\renewcommand{\theassertion}{1}
\begin{assertion}\cite{kuipersPolynomialTimeAlgorithm1996}
	Let $(N, v)$ be a convex game, let $U \subsetneq N$ be such that \mbox{$y(U)=z(U)$} $\forall y, z \in \leastcore(v)$, and let $x \in \core(v)$.
	Then for all $S \subsetneq U, S \neq  \emptyset$ we have
	\begin{equation*}
		\max_{Q\subseteq N\setminus U} \bigl(v(S\cup Q) - x(Q)\bigr) = \max\{v(S), \; v(S \cup (N\setminus U)) - x(N\setminus U)\}.
	\end{equation*}
\end{assertion}
\addtocounter{assertion}{-1}

Consider the game $(N,v)$ with $N=\{1,2,3,4\}$ and $v$ defined as:
\begin{equation*}
	v(S)=\begin{cases}
		3 \qquad &\text{if }S=\{1,2\},\{2,3\},\{2,3,4\},\\
		6 \qquad &\text{if }S=\{1,2,3\},\{1,2,4\},\\
		10 \qquad &\text{if }S=N,\\
		0 \qquad &\text{otherwise}.
	\end{cases}
\end{equation*}
Numerical computations show that $(N,v)$ is convex, and that its nucleolus is $\eta(N,v)=(\frac{5}{2}, \frac{7}{2},2,2)$, the least core value is $\varepsilon^*=2$, and both $\{3\}$ and $\{4\}$ are essential coalitions. 
For completeness, we provide a proof of these facts in Claim~\ref{cl:1} below.

Let us now prove that $(N,v)$ yields a counterexample to the Incorrect Theorem~\ref{prop:kuiperswrong}. 
Let  $U=\{1,2\}$ and $x= (\frac{5}{2}, \frac{7}{2},2,2)$.  
Since $(N,v)$ is convex and $x$ is its nucleolus, $x\in\core(v)$. 
Besides, $\{3\}$ and $\{4\}$ are essential coalitions and $\varepsilon^*=2$, therefore $y_3=y_4=2$ for all $y\in\leastcore(v)$. 
In particular, we have $y(U)=z(U)=10-4=6$ for all $y,z\in \leastcore(v)$. 
Now, for $S=\{2\}$, we have:
\begin{equation*}
    \begin{alignedat}{1}
v(\{2\}\cup \emptyset)-\eta(\emptyset)&=0-0=0,\\
v(\{2\}\cup \{3\})-\eta(\{3\})&=3-2=1,\\
v(\{2\}\cup \{4\})-\eta(\{4\})&=0-2=-2,\\
v(\{2\}\cup \{3,4\})-\eta(\{3,4\})&=3-4=-1.
    \end{alignedat}
\end{equation*}

It follows that $\max_{Q \subseteq N\setminus U} v(S\cup Q)-x(Q)=1\neq \max\{v(S),v(S\cup N\setminus U)-x(N\setminus U)\}=0$,  contradicting the conclusion of Incorrect Theorem~\ref{prop:kuiperswrong}.

\begin{claim}\label{cl:1}
Let $(N,v)$ be a game defined as above. Then the following hold:
\begin{enumerate}
    \item\label{it:kuipers-monotone} $v$ is monotone non-decreasing, i.e.\ $v(S)\leq v(T)$ for $S\subseteq T\subseteq N$. 
    \item\label{it:kuipers-convex} $(N,v)$ is convex.
    \item\label{it:kuipers-lcvalue} The least core value of $(N,v)$ is $\varepsilon^*=2$
    \item\label{it:kuipers-essentialcoalitions} $\{3\}$ and $\{4\}$ are essential coalitions.
    \item\label{it:kuipers-nucleolus} The nucleolus of $(N,v)$ is $\eta= \eta(N,v)=(\frac{5}{2}, \frac{7}{2},2,2)$
\end{enumerate}
\end{claim}
\proof
\ref{it:kuipers-monotone} We only need to check monotonicity for $S \subsetneq T \subseteq N$ such that $|T|=|S|+1$ and $|S|\leq 3$. 
The property is trivially satisfied for $S$ with $|S|\in \{0,1\}$,  $|S|=2$ and $v(S)=0$, or $|S|=3$. 
Now consider a set $S$ such that $\abs{S}=2$ and $v(S)\neq 0$, that is $\{1,2\}$ or $\{2,3\}$. 
Then $v(S)=3$, and any set $T$ with $\abs{T}=3$ containing $S$ has value $3$ or $6$. 

\ref{it:kuipers-convex} It is enough to prove that $v(S\cup\{i\})-v(S)\leq v(T\cup \{i\})-v(T)$ for all $S\subsetneq T \subsetneq N$ with $|T|=|S|+1$, and $i\not\in T$. 
In particular, since $|N|=4$ we have $|S|\leq2$. 
Let therefore $i\in N$ and $S\subseteq N\setminus\{i\}$ with $|S|\leq2$. 
Observe that, for any $U\subseteq N\setminus\{i\}$, $v(U\cup\{i\})-v(U)$ takes values in $\{0,3,4,6,7,10\}$.
However, since $|S|\leq2$, $v(S\cup\{i\})-v(S)\in\{0,3,6\}$.
We now enumerate over all the possible values for $v(S\cup\{i\})-v(S)$. 
\begin{itemize}
\item If $v(S\cup\{i\})-v(S) = 0$, then  the thesis $v(T\cup\{i\})-v(T) \geq 0$ follows from monotonicity, as $v(T\cup\{i\})\geq v(T)$. This case includes $S=\emptyset$, so from here onward we will assume $|S|\geq1$.

\item If $v(S\cup\{i\})-v(S)=6$, then $|S\cup\{i\}|=3$ and $i\in\{1,2\}$. It follows that $T\cup \{i\}=N$ and $T\in\{\{2,3,4\},\{1,3,4\}\}$. But then $v(T\cup\{i\})-v(T)\geq 7\geq6$.

\item If $v(S\cup\{i\})-v(S)=3$, assume by contradiction that $v(T\cup \{i\})-v(T)=0$, as all other possible values would imply that $(N,v)$ is convex.
Then $v(T\cup \{i\})=v(T)$ and it follows that either $v(T)=v(T\cup\{i\})=0$, or $T=\{2,3\}$ and $T\cup\{i\}=\{2,3,4\}$, since $|T|>|S|\geq1 \implies |T|\geq 2$. In the former case, we have $v(S)=v(T)=0$ and $v(S\cup\{i\})=v(T\cup\{i\})=0$ by monotonicity, and thus $v(S\cup\{i\})-v(S)=v(T\cup\{i\})-v(T)=0$, a contradiction. In the latter case, we have $S\in \{\{2\},\{3\}\}$ and $i=4$, but it implies $v(S\cup\{i\})=v(S)=0$ and thus $v(S\cup\{i\})-v(S)=0$, again a contradiction. 
\end{itemize}

\ref{it:kuipers-lcvalue}-\ref{it:kuipers-essentialcoalitions} 
For any feasible solution $(x,\excess)$ to (\ref{eq:least-core}), we have $x(S)\geq v(S)+\excess$ for any $S\subsetneq N, S\neq \emptyset$ and $x(N)=v(N)$. 
In particular, we have $x(\{1,2,3\}) \geq v(\{1,2,3\}) + \excess = 6 + \excess$,  $x(\{4\}) \geq v(\{4\}) + \excess=\excess$, and $x(\{1,2,3,4\})=v(\{1,2,3,4\})=10$. 
This implies that $10=x(N)=x(\{1,2,3\})+x(\{4\})\geq 6+2 \excess$ and therefore $\excess^* \leq 2$. Now, the point $(x=(\frac{5}{2}, \frac{7}{2},2,2),\excess=2)$ is feasible for (\ref{eq:least-core}), proving that $\optexcess=2$.
Furthermore, we have that $\{4\}$ is an essential coalition: the solution 
$\mu_{\{4\}}=-\frac{1}{2}$, $\mu_{\{1,2,3\}}=-\frac{1}{2}$, $\mu_N=\frac{1}{2}$, and $\mu_S=0$ otherwise is indeed an optimal solution to (\ref{eq:least-core-dual}), as it is feasible and has value $\optexcess$. 
Similar reasoning with $\{1,2,4\}$ and $\{3\}$  shows that also $\{3\}$ is an essential coalition.

\ref{it:kuipers-nucleolus} 
Since both $\{3\}$ and $\{4\}$ are essential coalitions, any solution $x\in \leastcore(v)$ has $x_3=x_4=2$.
This holds in particular for the nucleolus $\eta$, hence $\eta_3=\eta_4=2$.  We now consider the reduced game $(\{1,2\},v_{\{1,2\},\eta})$. 
By definition of reduced game (see (\ref{eq:reducedxcore})), we have: $v_{\{1,2\},\eta}(\emptyset)=0$, $v_{\{1,2\},\eta}(\{1,2\})=6$
$v_{\{1,2\},\eta}(\{1\})=\max \{v(\{1\}), v(\{1,3\})-\eta(\{3\}), v(\{1,4\})-\eta(\{4\}),v(\{1,3,4\})-\eta(\{3,4\})\}=0$,  
and $v_{\{1,2\},\eta}(\{2\})=\max \{v(\{2\}), v(\{2,3\})-\eta(\{3\}), v(\{2,4\})-\eta(\{4\}),v(\{2,3,4\})-\eta(\{3,4\})\}=1$.

The nucleolus of the reduced game $(\{1,2\},v_{\{1,2\},\eta})$ is the restriction of the nucleolus $\eta$ to $\{1,2\}$ by Proposition~\ref{prop:reducednucleolus}. 
For a two-player game, we have an explicit expression for the prenucleolus (which is the same as the nucleolus for convex games), see Proposition~\ref{prop:nucleolus2}. It follows that: 
\begin{equation*}
    \begin{alignedat}{1}
\eta_1&=v_{\{1,2\},\eta}(\{1\}) + \frac{v_{\{1,2\},\eta}(\{1,2\}) - v_{\{1,2\},\eta}(\{1\}) -v_{\{1,2\},\eta}(\{2\})}{2}=\frac{5}{2},\\
\eta_2&=v_{\{1,2\},\eta}(\{2\}) +\frac{v_{\{1,2\},\eta}(\{1,2\}) - v_{\{1,2\},\eta}(\{1\}) -v_{\{1,2\},\eta}(\{2\})}{2}=\frac{7}{2}.
    \end{alignedat}
\end{equation*}
\endproof

\begin{proposition}\label{prop:nucleolus2}
	Let $(N,v)$ be a game, with $\abs{N}=2$.
	The prenucleolus $\eta$ of  $(N,v)$ is the only point in $\leastcore(N,v)$ and is given by:
	\begin{equation*}
		\eta_j=v(\{j\})+\frac{1}{2}\Bigl(v(N)-\sum_{i \in N}v(\{i\})\Bigr) \qquad \forall j \in N.
	\end{equation*}
\end{proposition}
\proof
	We assume~$N=\{1,2\}$. The least core value $\optexcess$ is given by the following linear program:   
\begin{equation*}
    \begin{alignedat}{1}
 \optexcess = \; \text{maximize} \quad &\phantom{x_1+x_2}\mathllap{\excess}\\ 
  \text{subject to} \quad &\phantom{x_1+x_2}\mathllap{x_1} \ge v(\{1\})+\varepsilon\nonumber\\
 & \phantom{x_1+x_2}\mathllap{x_2} \ge v(\{2\})+\varepsilon \\
   & x_1+x_2=v(N).
    \end{alignedat}
\end{equation*}

	For any feasible $(x,\excess)$ we have $x_1+x_2=v(N) \ge v(\{1\})+v(\{2\})+2\varepsilon$, which implies
    \[\optexcess \le \frac{1}{2}\Bigl(v(N)-\bigl(v(\{1\})+v(\{2\})\bigr)\Bigr)=\bar{\varepsilon}.\]
	But $(\bar{x},\bar{\excess})$ with $\bar{x}$ defined as follows  
	\[\bar{x}_j=v(\{j\})+\frac{1}{2}\Bigl(v(N)-\bigl(v(\{1\})+v(\{2\})\bigr)\Bigr) =v(\{j\})+\bar{\excess} \qquad  \text{for }j =1,2\]
    is a feasible solution, and
    thus $\optexcess=\bar{\excess}$. 
    Moreover, $\{1\}$ and $\{2\}$ are essential coalitions, as $\mu_{\{1,2\}}=\frac{1}{2}$, $\mu_{\{1\}}=\mu_{\{2\}}=-\frac{1}{2}$  is a dual optimal solution (it is dual feasible and has value $\optexcess$). 
    This proves that the least core is a singleton and coincide with the prenucleolus and $\bar{x}$.
\endproof

\newpage
\section{Proofs}\label{sec:appendix_proofs}
\begin{proposition}\label{prop:reducedxcore}
	Let $(N,v)$ be a game, and ${x} \in \core(v)$.
	Then $\forall S \subsetneq N$, $S \neq \emptyset$ the reduced game $(S, v_{S, {x}})$ is given by:
	\[v_{S, {x}}(T)=\max_{R \subseteq N\setminus S}(v(T \cup R)-{x}(R))\qquad \forall T\subseteq S.\]
\end{proposition}

\proof
	We only need to check the equation when $T=\emptyset$ and $T=S$, as for $T \neq \emptyset, S$ this holds by definition.
	If $T=\emptyset$, we have:
	\[
	\max_{R \subseteq N\setminus S}\left(v(T \cup R)-{x}(R)\right)=\max_{R \subseteq N\setminus S}(v(R)-{x}(R)) \le 0,
	\]
	where the final inequality derives from ${x}(R)\ge v(R)$ $\forall R \subseteq N$, as ${x} \in \core(v)$.
	But the maximum is $0$ as $v(\emptyset)-{x}(\emptyset)=0$, which is also $v_{S, x}(\emptyset)$.
	If $T=S$, we want to prove that $\forall R \subseteq N \setminus S$:
	\[
		v_{S, {x}}(S)=v(N)-{x}(N \setminus S) 
		\ge v(S \cup R)-{x}(R)
		=v(S \cup R)-{x}(S \cup R)+{x}(S).
	\]
	The thesis is equivalent to checking
	\[v(N)-{x}(N \setminus S) -{x}(S)\ge v(S \cup R)-{x}(S \cup R),\]
	but the LHS is equal to $0$ (since $x\in \core(v))$, so this reduces to
	\[{x}(S \cup R)\ge v(S \cup R),\]
	that is satisfied since ${x} \in \core(v)$.
\endproof

\begin{proposition}\label{prop:reducedconvex}
	Let $(N,v)$ be a  convex game.
	For all $S \subsetneq N$, $S\neq \emptyset$ and $x \in \core(v)$, the reduced game $(S, v_{S, x})$ is convex.
\end{proposition}
\proof
	For simplicity, let us denote $v'=v_{S, x}$.
	Using Proposition~\ref{prop:reducedxcore}, we consider for all $T \subseteq S$:
    \begin{equation*}
        \begin{alignedat}{1}
          Q^*_T\coloneqq&\text{arg}\max_{Q \subseteq N\setminus S}(v(T \cup Q)-x(Q))\\
		v'(T)=&v(T \cup Q^*_T)-x(Q^*_T).  
        \end{alignedat}
    \end{equation*}
	We show the supermodularity of $v'$ for $T,R \subseteq S$:
    \begin{equation*}
        \begin{alignedat}{1}
		&v'(T)+v'(R) = v(T \cup Q^*_T)-x(Q^*_T)+v(R \cup Q^*_R)-x(Q^*_R)\\
		 \le& v((T \cup Q^*_T)\cup (R \cup Q^*_R))+v((T \cup Q^*_T)\cap (R \cup Q^*_R))
		-x(Q^*_T\cup Q^*_R)-x(Q^*_T\cap Q^*_R)\\
		 = &v((T \cup R)\cup (Q^*_T\cup Q^*_R))-x(Q^*_T\cup Q^*_R)
		+v((T \cap R)\cup (Q^*_T\cap Q^*_R))-x(Q^*_T\cap Q^*_R)\\
		\le& v'(T \cup R)+v'(T \cap R).
        \end{alignedat}
    \end{equation*}
The first inequality comes from the supermodularity of $v$.
In the third line, equality holds as $T,R \subseteq S$ and $Q^*_T, Q^*_R \subseteq N \setminus S$.
The last inequality derives from the definition of $v'$.
\endproof

\begin{proposition}\label{prop:reducedmax}
	Let $(N,v)$ be a  convex game, $x \in \leastcore(v)$, and  $S$ an essential coalition. 
	\begin{enumerate}
		\item \label{i:reducedmax}The reduced game $(S, v_{S, x})$ is given by:
		\begin{equation*}
			v_{S, x}(T)=\max \{v(T),\; v(T \cup (N \setminus S))-x(N \setminus S)\} \qquad \forall T \subseteq S.
		\end{equation*}
		
		\item  \label{ii:reducedmax} The reduced game $(N \setminus S, v_{N \setminus S, x})$ is given by:
		\begin{equation*}
			v_{N \setminus S, x}(T)=\max \{v(T), \;v(T \cup S)-x(S)\} \qquad \forall T \subseteq N \setminus S.
		\end{equation*}
	\end{enumerate}
\end{proposition}

\proof
\ref{i:reducedmax} Since $S$ is an essential coalition, we have $S\neq\emptyset,N$.
	 Consider $T \subseteq S$ and define
	\[Q^*\coloneqq\mathop{\mathrm{arg}}\max_{Q \subseteq N\setminus S, Q \neq N\setminus S}(v(T \cup Q)-x(Q)).\]
	From Proposition~\ref{prop:reducedxcore} we obtain
    \[
    v_{S, x}(T)=\max\{v(T \cup Q^*)-x(Q^*), \; v(T \cup (N\setminus S))-x(N\setminus S)\}.
    \]
    It thus suffices to prove that $v(T \cup Q^*)-x(Q^*)=v(T)$. Let $\optexcess$ be the least core value. 
    Remember that $x(S)=v(S)+\optexcess$ for an essential coalition $S$.
    \begin{equation*}
        \begin{alignedat}{1}
		v(T)-x(T)-\varepsilon^*
		&=v(T)-x(T)+v(S)-x(S)\\
		&\le v(T \cup Q^*)-x(Q^*)-x(T)+v(S)-x(S)\\
		&\le v(S \cup Q^*)-x(Q^*)-x(T)+v(T)-x(S)\\
		&= v(T)-x(T)+(v(S \cup Q^*)-x(S \cup Q^*))\\
		&\le v(T)-x(T)-\varepsilon^*.
        \end{alignedat}
    \end{equation*}
	The first equality comes from $S$ being an essential coalition.
	To pass to the second line we used the definition of $Q^*$, in particular $v(T \cup Q^*)-x(Q^*) \ge v(T \cup \emptyset)-x(\emptyset)=v(T)$.
	For the third line, we used the supermodularity of $v$ (with respect to the coalitions $T \cup Q^*$ and $S$), and $T\subseteq S$ and $Q^*\subsetneq N \setminus S$, implying 
    $v(T \cup Q^*)+v(S) \le v(T \cup Q^* \cup S) + v((T\cup Q^*)\cap S)=v(S \cup Q^*)+v(T)$. 
    The fourth line comes from $S\cap Q^*=\emptyset$.
	The last inequality derives from 
	$x \in \leastcore(v)$ and $S\cup Q^*\neq \emptyset,N$.
	Since the first and the last terms are equal, we have equalities all along, in particular:
        \begin{equation*}
        \begin{alignedat}{1}
		v(T)-x(T)+v(S)-x(S)&= v(T \cup Q^*)-x(Q^*)-x(T)+v(S)-x(S)\\
		\implies \quad v(T)&= v(T \cup Q^*)-x(Q^*).
        \end{alignedat}
        \end{equation*}

	\ref{ii:reducedmax} Since $x \in \leastcore(v)$, we have
	$\varepsilon^* \le x(Q)-v(Q)$ $\forall Q \subsetneq N$, $Q \neq \emptyset$.
	Let $T\subseteq N \setminus S$, then $\forall Q \subseteq S, Q \neq \emptyset$ we have:
	\[
		v(T \cup Q)-x(Q)+ \varepsilon^*
		\le v(T \cup Q)-v(Q)
		\le v(S \cup T)-v(S)=v(S \cup T)-x(S)+\varepsilon^*.
	\]
	Where the second inequality derives from the supermodularity of $v$ (with respect to the coalitions $T \cup Q$ and $S$) and $T \subseteq N \setminus S$, $Q\subseteq S$, implying $v(T \cup Q)+v(S) \le v(T \cup Q \cup S) + v((T \cup Q)\cap S) = v(T \cup S) +v(Q)$.
	From above we have:
    \begin{equation*}
        \begin{alignedat}{2}
		&v(T \cup Q)-x(Q) \le v(T \cup S)-x(S) &\forall Q \subseteq S, Q \neq \emptyset\\
		\implies &\max_{Q \subseteq S, Q \neq \emptyset} (v(T \cup Q)-x(Q)) = v(T \cup S)-x(S)\\
		\implies &\max_{Q \subseteq S} (v(T \cup Q)-x(Q)) = \max \{v(T \cup S)-x(S),\; v(T)\}.
        \end{alignedat}
        \end{equation*}
\endproof

\begin{lemma}[{\cite[Theorem 2.1]{schrijverCombinatorialOptimizationPolyhedra2004}}]
\label{refinement-sum-decrease}
Let $X,Y$ be subsets of a set $N$ with $X\not\subseteq Y$ and $Y\not\subseteq X$. Then
\[
\left|X\right|\left|N\setminus X\right|+\left|Y\right|\left|N\setminus Y\right|>\left|X\cap Y\right|\left|N\setminus\left(X\cap Y\right)\right|+\left|X\cup Y\right|\left|N\setminus\left(X\cup Y\right)\right|.
\]
\end{lemma}

\proof
Let $\alpha\coloneqq\left|X\cap Y\right|,\beta\coloneqq\left|X\setminus Y\right|,\gamma\coloneqq\left|Y\setminus X\right|$
and $\delta\coloneqq\left|N\setminus\left(X\cup Y\right)\right|$. Then
\begin{equation*}
        \begin{alignedat}{1}
\left|X\right|\left|N\setminus X\right|+\left|Y\right|\left|N\setminus Y\right| & =\left(\alpha+\beta\right)\left(\gamma+\delta\right)+\left(\alpha+\gamma\right)\left(\beta+\delta\right)\\
 & =2\alpha\delta+2\beta\gamma+\alpha\gamma+\beta\delta+\alpha\beta+\gamma\delta,
 \end{alignedat}
        \end{equation*}
and
\begin{equation*}
        \begin{alignedat}{1}
\left|X\cap Y\right|\left|N\setminus\left(X\cap Y\right)\right|+\left|X\cup Y\right|\left|N\setminus\left(X\cup Y\right)\right| & =\alpha\left(\beta+\gamma+\delta\right)+\left(\alpha+\beta+\gamma\right)\delta\\
 & =2\alpha\delta+\alpha\gamma+\beta\delta+\alpha\beta+\gamma\delta.
  \end{alignedat}
        \end{equation*}
The thesis follows from the fact that $\beta\gamma>0$.
\endproof

\begin{proposition}
\label{prop:convexity-superadditive-cover-near-convex-games}Let
$v\colon\powerset V\rightarrow\mathbb{R}$ be an intersecting supermodular
set function, then its superadditive cover $\supaddcover v$ is supermodular on $V$.  Moreover, $\supaddcover v(S)$ and a proper partition $\mathcal T$ of $S$ such that $\supaddcover v\left(S\right)=\sum\nolimits_{T\in\mathcal{T}}v\left(T\right)$ can be computed in (oracle) strongly polynomial time for all $S\subseteq V$.
\end{proposition}

\proof
Let $S,T\subseteq V$. 
If $S\cap T=\emptyset$, then $\supaddcover v(S)+\supaddcover v(T) \le \supaddcover v(S\cup T) =\supaddcover v(S \cup T) +\supaddcover v(S \cap T)$,
since $\supaddcover v\left(S\cap T\right)=0$ and $\supaddcover v$ is superadditive. 
If $S\cap T\neq\emptyset$,
then let $\mathcal{P},\mathcal{Q}$ be the two proper partitions of $S,T$
respectively for which 
\[
\supaddcover v\left(S\right) =\sum\limits_{P\in\mathcal{P}}v\left(P\right),\qquad \supaddcover v\left(T\right) =\sum\limits_{Q\in\mathcal{Q}}v\left(Q\right),
\]
and let $\mathcal{F}$ be the family formed by sets in both $\mathcal{P}$
and $\mathcal{Q}$ with repetition. 
The family $\mathcal{F}$ can be successively refined to guarantee that, for all $X,Y\in\mathcal{F}$, we have $X\subseteq Y$, $Y\subseteq X$ or $X\cap Y=\emptyset$, by replacing iteratively 
any $X,Y\in\mathcal{F}$ such that $X\cap Y\neq\emptyset$ and $X\nsubseteq Y\nsubseteq X$
with $X\cup Y,X\cap Y$. At each step, the sum $\sum_{U\in\mathcal{F}}v\left(U\right)$
does not decrease (due to the intersecting supermodularity of $v$),
while Lemma~\ref{refinement-sum-decrease} grants that the sum $\sum_{U\in\mathcal{F}}\left|U\right|\left|V\setminus U\right|$ always
decreases, proving that the refinement procedure eventually terminates.

Note that we also preserve that each element of $S\cap T$ is contained in exactly two sets of $\mathcal{F}$, while the remaining elements of $(S\cup T)\setminus(S\cap T)$ are each contained exactly once in a set of $\mathcal{F}$.
The family $\mathcal{M}$ of the inclusion-wise maximal sets of $\mathcal{F}$ (at the end of the refinement procedure) is therefore a partition of $S\cup T$, as the elements of $S \cup T$ are covered by at least an element of $\mathcal{M}$, and the intersection of the inclusion-wise maximal elements is empty.
The family $\mathcal{N}\coloneqq\mathcal{F}\setminus\mathcal{M}$ partitions the set $S\cap T$, as the elements of $S \cap T$ are covered exactly once, while the elements of $(S \cup T) \setminus (S\cap T)$ are not covered.
The supermodularity of $\supaddcover v$ follows from
\[
\supaddcover v(S)+\supaddcover v(T)=\sum\limits_{P\in\mathcal{P}}v(P)+\sum\limits_{Q\in\mathcal{Q}}v(Q)\leq\sum\limits_{M\in\mathcal{M}}v(M)+\sum\limits_{N\in\mathcal{N}}v(N)\leq\supaddcover v(S\cup T)+\supaddcover v(S\cap T).
\]


Algorithm~\ref{alg:superadditive-cover} allows us to compute in strongly polynomial time
the value of the superadditive cover $\supaddcover v$ for a given
$S\subseteq V$, $S \neq \emptyset$ when $v$
is intersecting supermodular on $S$ (while for $S=\emptyset$, $\supaddcover v = 0$ by definition); the correctness of algorithm is proved
below.

\begin{algorithm}[h]
\begin{algorithmic}[1]
\begin{inputblock}
A set $S\subseteq V$, $S \neq \emptyset$ and a set function $v\colon\powerset V\rightarrow\mathbb{R}$ intersecting
supermodular on $S$

\end{inputblock}

\begin{outputblock}
The value of $\supaddcover v\left(S\right)$ and a  proper partition $\mathcal T$ of $S$ such that $\supaddcover v\left(S\right)=\sum\nolimits_{T\in\mathcal{T}}v\left(T\right)$

\end{outputblock}

\State{Assume \textsc{w.l.o.g.}~that $S=\{1,...,k\}$} 

\State{$x_i\algset0$\;$\forall i \in S$}

\For{$i= 1,\ldots,k $}{}

\State{\noindent$T_{{i}}\algset\argmax\left\{ v\left(U\right)-x\left(U\setminus\{i\}\right)\setdef U\subseteq\left\{ {1},\ldots,{i}\right\}, {i}\in U \right\} $\label{line:superadditive-cover-construction-1}}

\State{\noindent$x_{{i}}\algset v\left(T_{{i}}\right)-x\left(T_{{i}}\setminus\{i\}\right)$\label{line:superadditive-cover-construction-2}}

\EndFor{}

\State{$\mathcal{T}\algset\left\{ T_{i}\setdef i= 1,\ldots,k \right\} $\label{line:superadditive-cover-family-def}}

\For{$T',T''\in\mathcal{T}$ with $T'\neq T''$, $T'\cap T''\neq\emptyset$}{}

\State{$\mathcal{T}\algset\left(\mathcal{T}\setminus\left\{ T',T''\right\} \right)\cup\left\{ T'\cup T''\right\} $\label{line:superadditive-cover-refinement}}

\EndFor{}

\returnstmt{$\sum\nolimits_{T\in\mathcal{T}}v\left(T\right)$, $\mathcal T$}

\end{algorithmic}

\caption{\label{alg:superadditive-cover}Superadditive cover}
\end{algorithm}

In the notation of Algorithm~\ref{alg:superadditive-cover}, we now prove that we have $\supaddcover v\left(S\right)=\sum\nolimits_{T\in\mathcal{T}}x\left(T\right)$ at the end of the algorithm.
Note that $i \in T_i$ $\forall i = 1, \ldots, k$ by line~\ref{line:superadditive-cover-construction-1}.
From the definition of $\mathcal{T}$ of line~\ref{line:superadditive-cover-family-def} and from line~\ref{line:superadditive-cover-construction-2}, it holds that $\bigcup_{T\in\mathcal{T}}T=S$
and 
\begin{equation}\label{eq:proof_alg3_1}
    x\left(T\right)=v\left(T\right)\qquad\forall T\in\mathcal{T}.
\end{equation}
Now, let $R\subseteq S$, and let $h=\max \left\{i \setdef i\in R\right\}$.
From lines~\ref{line:superadditive-cover-construction-1}-\ref{line:superadditive-cover-construction-2}
it follows that $x_{{h}}\geq v\left(R\right)-x\left(R\setminus\left\{ {h}\right\} \right)$
and therefore 
\begin{equation}\label{eq:proof_alg3_2}
    x\left(R\right)\geq v\left(R\right) \qquad \forall R\subseteq S.
\end{equation}
We prove by induction that $x\left(T\right)= v\left(T\right)$ remains true over all elements of $\mathcal{T}$ after the iterative refinement of line~\ref{line:superadditive-cover-refinement}. 
Assume inductively that $x\left(T\right)=v\left(T\right)$
$ \forall T\in \mathcal{T}$, and consider $T', T'' \in \mathcal{T}$ with $T'\neq T''$ and $T'\cap T''\neq\emptyset$.
Using $x\left(T'\cap T''\right)\geq v\left(T'\cap T''\right)$ (from (\ref{eq:proof_alg3_2})),
$T', T'' \in \mathcal{T}$,
$x\left(T'\right)=v\left(T'\right)$, \mbox{$x\left(T''\right)=v\left(T''\right)$} (from (\ref{eq:proof_alg3_1})),
and thanks to the intersecting supermodularity of $v$ on $S$, we have
\[
v\left(T'\cup T''\right)\geq v\left(T'\right)+v\left(T''\right)-v\left(T'\cap T''\right)\geq x\left(T'\right)+x\left(T''\right)-x\left(T'\cap T''\right)=x\left(T'\cup T''\right),
\]
but since also $x\left(T'\cup T''\right) \geq v\left(T'\cup T''\right)$,
it follows that $v\left(T'\cup T''\right)=x\left(T'\cup T''\right)$. Therefore
at each refinement step the equality $x\left(T\right)=v\left(T\right)$
remains true over all elements $T\in\mathcal{T}$ , and by construction
the relation $\bigcup_{T\in\mathcal{T}}T=S$ is similarly preserved.
The final refinement of $\mathcal{T}$ is composed of non-empty disjoint sets
and is therefore a proper partition of $S$, for which $\supaddcover v\left(S\right)=\sum\nolimits_{T\in\mathcal{T}}x\left(T\right)$
since for any $\mathcal{P}$ proper partition of $S$
\[
\supaddcover v\left(S\right)\geq\sum\nolimits_{T\in\mathcal{T}}v\left(T\right)=\sum\nolimits_{T\in\mathcal{T}}x\left(T\right)=x\left(S\right)=\sum\nolimits_{P\in\mathcal{P}}x\left(P\right)\geq\sum\nolimits_{P\in\mathcal{P}}v\left(P\right),
\]
and thus $\supaddcover v\left(S\right)\geq\sum\nolimits_{T\in\mathcal{T}}v\left(T\right)\geq\max\limits_{\text{\ensuremath{\mathcal{P}} prop.\ part.\ of \ensuremath{S}}}\sum\nolimits_{P\in\mathcal{P}}v\left(P\right)=\supaddcover v\left(S\right)$ and equality holds all along.

Observe that $f(T)\coloneqq v(T \cup \{i\})-x(T)$ is supermodular for $T \subseteq \{1, \ldots, i-1\}$, as $v$ is intersecting supermodular on $S$.
Thus
\[\max\left\{ v(U)-x(U\setminus\{i\})\setdef U\subseteq\left\{ {1},\ldots,{i}\right\}, {i}\in U \right\} = \max_{T \subseteq \{1, \ldots, i-1\}} f(T) \]
can be computed in (oracle) strongly polynomial time through submodular function minimization. 
We note that we can take care of the refinement steps directly in line~\ref{line:superadditive-cover-construction-1} using a similar trick as in~(\ref{eq:usualtrick}).
\endproof

\newpage
\section{Computing Least Core value and essential coalitions as a Discrete Newton Method}\label{sec:appendix_newton}
Algorithm~2 can be interpreted as a discrete Newton-Dinkelbach method; see \cite{goemans2017discrete}.
Fix $s\in N$ and let $\mathcal C$ consist of all triples $C=(S,\mathcal P,\mathcal Q)$, where $\varnothing\ne S\subseteq N\setminus\{s\}$ and $\mathcal P,\mathcal Q$ are partitions of $S$.
For each $C\in\mathcal C$, define
\[
a_C=|\mathcal Q|v(N)
    -\sum_{P\in\mathcal P}v(P)
    -\sum_{Q\in\mathcal Q}v(N\setminus Q),
\qquad
b_C=|\mathcal P|+|\mathcal Q|.
\]
Expanding the superadditive covers gives
\[
F(\varepsilon)
:=\min_{S\subseteq N\setminus\{s\}}
\bigl[-\overline v_\varepsilon(S)
      -\overline u_\varepsilon(S)\bigr]
=\min\left\{0,\min_{C\in\mathcal C}
(a_C-b_C\varepsilon)\right\}.
\]
The constant zero term corresponds to $S=\varnothing$.
As a finite minimum of affine functions with nonpositive
slopes, the function $F$ is concave, piecewise affine and nonincreasing.
By Theorem~3, $\varepsilon$ is feasible if and only if
$F(\varepsilon)=0$; thus $\varepsilon^*$ is its largest zero.

At each nonterminating iteration $i$, the minimizing set and the partitions retrieved by Algorithm~2 yield a certificate $C_i$ satisfying
\[
F(\varepsilon_i)=a_{C_i}-b_{C_i}\varepsilon_i<0.
\]
Consequently, $-b_{C_i}$ is a supergradient of $F$ at
$\varepsilon_i$, and the update in Algorithm~2 is exactly
the discrete Newton step
\[
\varepsilon_{i+1}
=\varepsilon_i-\frac{F(\varepsilon_i)}{-b_{C_i}}
=\frac{a_{C_i}}{b_{C_i}}.
\]
The standard discrete Newton slope-monotonicity argument (see, e.g., Lemma~1(iii) in \cite{goemans2017discrete}) implies that $b_{C_i}$ strictly decreases between consecutive
nonterminating iterations. 
Since $b_C\in\{2,\ldots,2n-2\}$, this also explains the bound of $2n-3$ updates and one final feasibility test in Proposition~5.

This interpretation concerns the parameter search and its iteration bound. 
The argument in Proposition~5 showing that the final certificate identifies coalitions tight at every least-core allocation remains necessary to establish the
essentiality of the returned coalitions, as required by the reduced-game recursion.


\global\long\def\powerset#1{2^{#1}}%

\global\long\def\setfamily{\mathcal{C}}%

\newcommandx\extendedpolymatroid[1][usedefault, addprefix=\global, 1=]{EP_{#1}}%

\newcommandx\basepolytope[1][usedefault, addprefix=\global, 1=]{B_{#1}}%

\global\long\def\setdef{\;\middle|\;}%

\global\long\def\partition{\mathcal{P}}%

\global\long\def\complementaryy#1{\lnot#1}%

\global\long\def\complementary#1{#1^{*}}%

\global\long\def\checkop{\mathop{\text{chk}}}%

\global\long\def\hatop{\mathop{\text{hat}}}%

\global\long\def\setdef{\;\middle|\;}%

\global\long\def\argmax{\mathop{\arg\max}}%
\global\long\def\argmin{\mathop{\arg\min}}%

\global\long\def\algset{\leftarrow}%

\global\long\def\Ns{N_{s}}%

\newcommandx\restrictedfamily[1][usedefault, addprefix=\global, 1=S]{\mathcal{C}_{#1}}%

\newcommandx\restrictedfun[2][usedefault, addprefix=\global, 1=S, 2=f]{#2_{#1}}%

\newcommandx\restrict[2][usedefault, addprefix=\global, 1=f, 2=\restrictedfamily]{\left.#1\right|_{#2}}%

\newcommandx\checktime[2][usedefault, addprefix=\global, 1=, 2=n]{\mathord{\mathrm{CHK}}_{#1}}%

\global\long\def\bigo{\mathop{O}}%
\global\long\def\bigomega{\mathop{\Omega}}%
\newcommandx\complvarp[1][usedefault, addprefix=\global, 1=n]{p_{#1}}%
\newcommandx\complvarq[1][usedefault, addprefix=\global, 1=n]{q_{#1}}%

\newpage
\section{Testing non-emptiness  of base polytopes associated with crossing submodular functions over families closed under complement}\label{app:C}

In this section, we show that Theorem~\ref{thm:empty} extends to testing the non-emptiness of base polytopes associated with crossing submodular functions on families that are closed under complement. We believe this result may be of independent interest to the submodular function optimization community. For completeness, we provide a self-contained proof and introduce first all necessary definitions and useful results in this broader setting.

\subsection{Preliminaries}

Let $N\subseteq \mathbb{N}$ be finite and non-empty. A family $\setfamily\subseteq\powerset N$ is a \emph{lattice family}
if for any $S,T\in\setfamily$ then $S\cap T,S\cup T\in\setfamily$. A set function $f\colon\setfamily\rightarrow\mathbb{R}$ is called
\emph{lattice submodular} if $f\left(S\right)+f\left(T\right)\geq f\left(S\cap T\right)+f\left(S\cup T\right)$ for all $S,T\in\setfamily$. A general lattice submodular function $f\colon\setfamily\rightarrow\mathbb{R}$
on the lattice family $\setfamily$ can be minimized in polytime by a reduction to
submodular function minimization \cite[§49.3]{schrijverCombinatorialOptimizationPolyhedra2004} (note that $\setfamily$ must be given by the smallest set $\underline{L}$ and the largest set $\bar{L}$ in $\setfamily$, together with the pre-order $\preccurlyeq$ defined by: $u\preccurlyeq v \iff$ each $U\in \setfamily$ containing $v$ also contains $u$).

A family $\setfamily\subseteq\powerset N$ is an \emph{intersecting
family} if for any $S,T\in\setfamily$ with $S\cap T\neq\emptyset$,
then $S\cap T,S\cup T\in\setfamily$. A set function $f\colon\setfamily\rightarrow\mathbb{R}$
is called \emph{intersecting submodular} if $f\left(S\right)+f\left(T\right)\geq f\left(S\cap T\right)+f\left(S\cup T\right)$ for all  $S,T\in\setfamily\enspace\text{with }S\cap T\neq\emptyset$.  A family $\setfamily\subseteq\powerset N$ is a \emph{co-intersecting
family} if for any $S,T\in\setfamily$ with $S\cup T\neq N$, then
$S\cap T,S\cup T\in\setfamily$. A set function $f\colon\setfamily\rightarrow\mathbb{R}$
is called \emph{co-intersecting submodular} if
$f\left(S\right)+f\left(T\right)\geq f\left(S\cap T\right)+f\left(S\cup T\right)
$ for all $S,T\in\setfamily\enspace\text{with }S\cup T\neq N$. A family $\setfamily\subseteq\powerset N$ is a \emph{crossing family}
if for any $S,T\in\setfamily$ with $S\cap T\neq\emptyset$ and $S\cup T\neq N$,
then $S\cap T,S\cup T\in\setfamily$. A set function $f\colon\setfamily\rightarrow\mathbb{R}$
is called \emph{crossing submodular} if 
$
 f\left(S\right)+f\left(T\right)\geq f\left(S\cap T\right)+f\left(S\cup T\right)
$ for all $S,T\in\setfamily\enspace\text{with }S\cap T\neq\emptyset,S\cup T\neq N$. 
For convenience, we additionally introduce the concept of complement
of a set function. For $f\colon\setfamily\rightarrow\mathbb{R}$ set function with $N\in\setfamily$, we call
$\complementary f\colon\complementary{\setfamily}\rightarrow\mathbb{R}$
its \emph{complement} function, defined as 
$\complementary f  \colon S\longmapsto f\left(N\setminus S\right)-f\left(N\right)$ where $ \complementary{\setfamily}  \coloneqq\left\{ N\setminus S\setdef S\in\setfamily\right\}$.

From the definition, it follows that
$f=\complementary{\left(\complementary f\right)}$ under the assumption
that $\emptyset\in\setfamily$ and $f\left(\emptyset\right)=0$, and simple set-theoretical considerations
imply that the submodularity classes of $f$ and $\complementary f$
are related in the following manner:

\begin{center}
\begin{tabular}{c|cccc}
$f$ & lattice & intersecting & co-intersecting & crossing\\
\hline 
$\complementary f$ & lattice & co-intersecting & intersecting & crossing\\
\end{tabular}
\par\end{center}

For any submodular function (and in general for any set function)
two polytopes may be defined as follows.
Given a set function $f\colon\setfamily\rightarrow\mathbb{R}$ on
$\setfamily\subseteq\powerset N$ we define the \emph{extended polymatroid}
associated with $f$ as
\[\extendedpolymatroid[f]\coloneqq\left\{ x\in\mathbb{R}^{|N|}\setdef x\left(S\right)\leq f\left(S\right)\;\forall S\in\setfamily\right\} ,
\]
and when $N\in\setfamily$, the \emph{base polytope} of $f$ as
\begin{align*}
    \basepolytope[f]
    \coloneqq  \;&\extendedpolymatroid[f]\cap\left\{ x\in\mathbb{R}^{|N|}\setdef x\left(N\right)=f\left(N\right)\right\} \\
= &\left\{ x\in\mathbb{R}^{|N|}\setdef x\left(S\right)\leq f\left(S\right)\;\forall S\in\setfamily\text{ and }x\left(N\right)=f\left(N\right)\right\}.
\end{align*}

Notice how in order
to have $\extendedpolymatroid[f]\neq\emptyset$ it must be that $\emptyset\notin\setfamily$
or $f\left(\emptyset\right)\geq0$. The converse is also true: if
either $\emptyset\notin\setfamily$ or $f\left(\emptyset\right)\geq0$
then $\extendedpolymatroid[f]\neq\emptyset$, since $x\in\mathbb{R}^{|N|}$
defined as $x_{s}\coloneqq-M$ for all $s\in N$ with $M>\max_{S\in\setfamily}\left|f\left(S\right)\right|$
is such that $x\in\extendedpolymatroid[f]$. Additionally, notice
how if $x'\in\extendedpolymatroid[f]$, then any $x\in\mathbb{R}^{|N|}$
with $x\leq x'$ (where the inequality is meant element-wise) is such
that $x\in\extendedpolymatroid[f]$. 

For a lattice submodular function with
either $\emptyset\notin\setfamily$ or $f\left(\emptyset\right)\geq0$
the base polytope $\basepolytope[f]$ is also always non-empty. Indeed,
a point $x\in\basepolytope[f]$ can be found by means of an adaptation of the \emph{greedy
algorithm} \cite[§49.3]{schrijverCombinatorialOptimizationPolyhedra2004}. The non-emptiness of $\basepolytope[f]$ is instead not granted in
the case of a function $f$ intersecting or crossing submodular on
$\setfamily$, as differently from the lattice submodular case it
may happen that $\sum_{P\in\mathcal{P}}f\left(P\right)<f\left(N\right)$
for some $\mathcal{P}\subseteq\setfamily$ proper partition of $N$,
and in this case
\[
x\in\extendedpolymatroid[f]\quad\implies\quad x\left(N\right)=\sum_{P\in\mathcal{P}}x\left(P\right)\leq\sum_{P\in\mathcal{P}}f\left(P\right)<f\left(N\right)\quad\implies\quad x\left(N\right)<f\left(N\right),
\]
implying $\basepolytope[f]=\emptyset$. This motivates the introduction of two new functions $\check{f},\hat{f}$, defined as follows. \\
Let $f\colon\setfamily\rightarrow\mathbb{R}$ be an intersecting submodular function on $\setfamily\subseteq\powerset N$. We define the function $\check{f}\colon\check{\setfamily}\rightarrow\mathbb{R}$
with $\check{\setfamily}\coloneqq\left\{ \bigcup\nolimits_{T\in\mathcal{T}}T\setdef\mathcal{T}\subseteq\setfamily\right\} $
as
\begin{align*}
\check{f}\left(S\right)\coloneqq\min\Big\{\sum\nolimits_{P\in\mathcal{P}}f\left(P\right)\;\Big|\;\mathcal{P}\subseteq\setfamily\text{ proper partition of \ensuremath{S}}\Big\}\tag{\textsc{check}}
\end{align*}
$\forall S\in\check{\setfamily}\setminus\left\{ \emptyset\right\} $
and $\check{f}\left(\emptyset\right)=0$ (we allow the union of zero sets $\mathcal{T}=\emptyset$, and correspondingly $\emptyset\in\check{\setfamily}$). 
Note how $\check{f}$ is well-defined, as $\setfamily$ is an intersecting family and thus closed under unions of intersecting sets, making $\check{\setfamily}$ equivalent to the family of disjoint unions of sets in $\setfamily$. \\
Let $f\colon\setfamily\rightarrow\mathbb{R}$ be a co-intersecting submodular function on $\setfamily\subseteq\powerset N$. In case $f\left(N\right)=0$, we define $\hat{f}\colon\hat{\setfamily}\rightarrow\mathbb{R}$
with $\hat{\setfamily}\coloneqq\left\{ \bigcap_{T\in\mathcal{T}}T\setdef\mathcal{T}\subseteq\setfamily\right\}$
as
\begin{align*}
\hat{f}\left(S\right)\coloneqq & \min\Big\{\sum\nolimits_{Q\in\mathcal{Q}}f\left(Q\right)\;\Big|\;\mathcal{Q}\subseteq\setfamily\text{ proper copartition of \ensuremath{N\setminus S}}\Big\}\tag{\textsc{hat}}
\end{align*}
$\forall S\in\hat{\setfamily}\setminus\left\{ N\right\} $ and $\hat{f}\left(N\right)=0$
(we allow the intersection of zero sets $\mathcal{T}=\emptyset$, and correspondingly $N\in\hat{\setfamily}$),
where $\mathcal{Q}$ is a \emph{(proper) copartition} of a set if
and only if $\left\{ N\setminus Q\setdef Q\in\mathcal{Q}\right\} $
is a (proper) partition of the same set. Similarly to before, $\hat{f}$ is well-defined as a consequence of $\setfamily$ being a co-intersecting family. \\
If $f\colon\setfamily\rightarrow\mathbb{R}$ is instead taken to be a crossing submodular function on $\setfamily\subseteq\powerset N$, the functions $\check{f}$ and $\hat{f}$ can still be defined by taking the domain families to be $\check{\setfamily}\coloneqq\left\{ \bigcup\nolimits_{T\in\mathcal{T}}T\setdef\mathcal{T}\subseteq\setfamily\right\} \setminus\left\{ N\right\} $ and $\hat{\setfamily}\coloneqq\left\{ \bigcap_{T\in\mathcal{T}}T\setdef\mathcal{T}\subseteq\setfamily\right\} \setminus\left\{ \emptyset\right\} $ respectively.

Dunstan proved in \cite{dunstanMatroidsSubmodularFunctions1976}
the following result on intersecting submodular functions, provided
also in \cite[Theorem 49.4]{schrijverCombinatorialOptimizationPolyhedra2004}.
\begin{proposition}
\label{prop:lattice-submodularity-of-f-check}If $f\colon\setfamily\rightarrow\mathbb{R}$
is a intersecting submodular function on $\setfamily\subseteq\powerset N$,
then $\check{f}\colon\check{\setfamily}\rightarrow\mathbb{R}$ is
a lattice submodular function on \textup{$\check{\setfamily}$.}
\end{proposition}

A similar result holds for crossing submodular functions with respect
to the hat operator (see \cite[Theorem 49.6]{schrijverCombinatorialOptimizationPolyhedra2004}).
\begin{proposition}
\label{prop:intersecting-submodularity-of-f-hat}If $f\colon\setfamily\rightarrow\mathbb{R}$
is a crossing submodular function on $\setfamily\subseteq\powerset N$,
$f\left(N\right)=0$, then $\hat{f}\colon\hat{\setfamily}\rightarrow\mathbb{R}$
is an intersecting submodular function on \textup{$\hat{\setfamily}$.}
\end{proposition}

We now show a polynomial-time algorithmic procedure, adapted from Schrijver
\cite[§49.7]{schrijverCombinatorialOptimizationPolyhedra2004}
and displayed in Algorithm~\ref{alg:f-check-computation}, for computing the
value of $\check{f}\left(S\right)$ for $S\in\check{\setfamily}$
given an intersecting submodular function $f\colon\setfamily\rightarrow\mathbb{R}$.
We will assume to have a membership oracle for
$\setfamily$, returning whether $S\in\setfamily$ or not for a given
set $S\subseteq N$.
\begin{algorithm}[H]
\begin{algorithmic}[1]
\begin{inputblock}
A function $f\colon\setfamily\rightarrow\mathbb{R}$ intersecting
submodular on $\setfamily\subseteq\powerset N$, a set $S\in\check{\setfamily}$, and a bound $M\geq\max\left\{ \left|f\left(U\right)\right|\setdef U\in\setfamily,\emptyset\neq U\subseteq S\right\} $

\end{inputblock}

\begin{outputblock}
The value of $\check{f}\left(S\right)$

\end{outputblock}

\State{Order the elements of $S$ arbitrarily as $s_{1},...,s_{k}$\label{line:f-check-init-ordering}}

\State{$x_{i}\algset-M\quad\forall i\in S$}

\For{$i\in\left\{ 1,...,k\right\} $}{}

\State{\noindent$T_{i}\algset\argmin\left\{ f\left(U\right)-x\left(U\setminus\left\{ s_{i}\right\} \right)\setdef U\in\setfamily,U\subseteq S,s_{i}\in U\right\} $\label{line:f-check-construction-1}}

\State{\noindent$x_{s_{i}}\algset f\left(T_{i}\right)-x\left(T_{i}\setminus\left\{ s_{i}\right\} \right)$\label{line:f-check-construction-2}}

\EndFor{}

\State{$\mathcal{T}\algset\left\{ T_{i}\setdef i=1,...,k\right\} $\label{line:f-check-family-def}}

\For{$T',T''\in\mathcal{T}$ with $T'\neq T''$, $T'\cap T''\neq\emptyset$}{\label{line:f-check-refinement-conditions}}

\State{$\mathcal{T}\algset\left(\mathcal{T}\setminus\left\{ T',T''\right\} \right)\cup\left\{ T'\cup T''\right\} $\label{line:f-check-refinement}}

\EndFor{}

\returnstmt{$\sum\nolimits_{T\in\mathcal{T}}f\left(T\right)$}

\end{algorithmic}

\caption{\label{alg:f-check-computation}Computation of $\check{f}$}
\end{algorithm}
The following Proposition~\ref{prop:f-check-computation} proves the correctness
of Algorithm~\ref{alg:f-check-computation}.
\begin{proposition}
\label{prop:f-check-computation}In the notation of Algorithm~\ref{alg:f-check-computation},
 $\check{f}\left(S\right)=\sum\nolimits_{T\in\mathcal{T}}f\left(T\right)=\sum\nolimits_{T\in\mathcal{T}}x\left(T\right)$.
\end{proposition}

\proof
Note first that $s_{i}\in T_{i}\subseteq S$ $\forall i=1,...,k$ by line~\ref{line:f-check-construction-1}
(where the minimum is over a non-empty family, as $S\in\check{\setfamily}$),
hence from the definition of $\mathcal{T}$ of line~\ref{line:f-check-family-def}
it holds that $\bigcup_{T\in\mathcal{T}}T=S$.
We claim that $x\left(U\right)\leq f\left(U\right)$ for all $U\in\setfamily$ with $\emptyset\neq U\subseteq S$
throughout the algorithm. Initially this holds since $x\left(U\right)=-M\left|U\right|\leq-M\leq f\left(U\right)$.
At iteration $i$, line~\ref{line:f-check-construction-2} only changes $x_{s_{i}}$,
and by line~\ref{line:f-check-construction-1} its new value satisfies
$x_{s_{i}}\leq f\left(U\right)-x\left(U\setminus\left\{ s_{i}\right\} \right)$
for all such $U$ with $s_{i}\in U$, so the claim is preserved.
Moreover, before line~\ref{line:f-check-construction-2} the claim gives
$f\left(U\right)-x\left(U\setminus\left\{ s_{i}\right\} \right)\geq x_{s_{i}}=-M$
for all such $U$ with $s_{i}\in U$. Hence, if such a $U$ satisfies $x\left(U\right)=f\left(U\right)$
before iteration $i$, the minimum of line~\ref{line:f-check-construction-1} equals $-M$,
$x_{s_{i}}$ keeps the value $-M$, and $x\left(U\right)=f\left(U\right)$ still holds after iteration $i$.
Since $x\left(T_{i}\right)=f\left(T_{i}\right)$ after iteration $i$ by line~\ref{line:f-check-construction-2},
and every later iteration $j$ preserves this equality (trivially if $s_{j}\notin T_{i}$, and by the previous argument if $s_{j}\in T_{i}$),
at the end of the algorithm $x\left(T\right)=f\left(T\right)$ $\forall T\in\mathcal{T}$.

Notice that the conditions of line~\ref{line:f-check-refinement-conditions}
guarantee that at each step of the iterative refinement of line~\ref{line:f-check-refinement}
the set $T'\cup T''$ (as well as the set $T'\cap T''$) is still
in $\setfamily$. We prove by induction that the inequality $x\left(T\right)\leq f\left(T\right)$
remains tight over all elements of $\mathcal{T}$ after each of these
steps. The induction hypothesis is $x\left(T'\right)=f\left(T'\right),x\left(T''\right)=f\left(T''\right)$
for $T',T''\in\mathcal{T}$ with $T'\neq T''$ and $T'\cap T''\neq\emptyset$.
Recalling that $x\left(T'\cap T''\right)\leq f\left(T'\cap T''\right)$,
and thanks to the intersecting submodularity of $f$, we have
\[
f\left(T'\cup T''\right)  \leq f\left(T'\right)+f\left(T''\right)-f\left(T'\cap T''\right)
  \leq x\left(T'\right)+x\left(T''\right)-x\left(T'\cap T''\right)
  =x\left(T'\cup T''\right),        
    \]
but since also $x\left(T'\cup T''\right)\leq f\left(T'\cup T''\right)$,
it follows $x\left(T'\cup T''\right)=f\left(T'\cup T''\right)$. Therefore
\mbox{$x\left(T\right)=f\left(T\right)$} remains true over all elements
$T\in\mathcal{T}$. By construction the relation $\bigcup_{T\in\mathcal{T}}T=S$
is also similarly preserved after each step. The final refinement
of $\mathcal{T}$ is composed of disjoint sets and is therefore a
partition of $S$.

To conclude, we have that $\check{f}\left(S\right)=\sum\nolimits_{T\in\mathcal{T}}x\left(T\right)$
since for any $\mathcal{P}$ proper partition of $S$
\[
\check{f}\left(S\right)\leq\sum\nolimits_{T\in\mathcal{T}}f\left(T\right)=\sum\nolimits_{T\in\mathcal{T}}x\left(T\right)=x\left(S\right)=\sum\nolimits_{P\in\mathcal{P}}x\left(P\right)\leq\sum\nolimits_{P\in\mathcal{P}}f\left(P\right),
\]
thus 
\[
\check{f}\left(S\right)\leq\sum\limits_{T\in\mathcal{T}}f\left(T\right)\leq\min_{\mathcal{P}\text{ prop.\ part.\ of }S}\sum\limits_{P\in\mathcal{P}}f\left(P\right)=\check{f}\left(S\right),
\]
 equality holding throughout.
\endproof

Algorithm~\ref{alg:f-check-computation}
may also be used to compute the value of $\check{f}\left(S\right)$
when $f$ is crossing (or co-intersecting) submodular on $\setfamily$,
if $S\neq N$. Indeed whenever $S\subsetneq N$, the restriction $\restrictedfun\coloneqq\restrict$
to the set $\restrictedfamily\coloneqq\setfamily\cap\powerset S=\left\{ T\in\setfamily\setdef T\subseteq S\right\} $
becomes intersecting submodular on its domain $\setfamily_{S}$, while
from the definition of the check operator replacing $f$ with $\restrictedfun$
does not change the value of $\check{f}\left(S\right)$, that is $\check{f}\left(S\right)=\check{\restrictedfun}\left(S\right)$.

Besides, assuming again \textsc{w.l.o.g.}~that $f\left(N\right)=0$,
Algorithm~\ref{alg:f-check-computation} also allows us to compute $\hat{f}$
thanks to the relation
\begin{equation*}
        \begin{alignedat}{1}
\hat{f}\left(S\right)= & \min\Big\{\sum\nolimits_{Q\in\mathcal{Q}}f\left(Q\right)\;\Big|\;\mathcal{Q}\subseteq\setfamily\text{ proper copartition of \ensuremath{N\setminus S}}\Big\}\\
= & \min\Big\{\sum\nolimits_{P\in\mathcal{P}}f\left(N\setminus P\right)\;\Big|\;\mathcal{P}\subseteq\complementary{\setfamily}\text{ proper partition of \ensuremath{N\setminus S}}\Big\}\\
= & \min\Big\{\sum\nolimits_{P\in\mathcal{P}}\left(\complementary f\left(P\right)+\cancel{f\left(N\right)}\right)\;\Big|\;\mathcal{P}\subseteq\complementary{\setfamily}\text{ proper partition of \ensuremath{N\setminus S}}\Big\}\\
= & \check{\left(\complementary f\right)}\left(N\setminus S\right).
\end{alignedat}
\end{equation*}
The intersecting submodularity of $\complementary f$ on $\complementary{\setfamily}$
required by the algorithm, becomes in this case equivalent to the
co-intersecting submodularity of $f$ on $\setfamily$.
As a consequence,
when $f$ is crossing (or intersecting) submodular on $\setfamily$,
$\hat{f}\left(S\right)$ can still be computed when $N\setminus S\neq N$
(i.e.~$S\neq\emptyset$) by computing $\check{\restrictedfun[N\setminus S][\complementary f]}\left(N\setminus S\right)$
instead.

The polynomial-time complexity of Algorithm~\ref{alg:f-check-computation} is
discussed in \S\ref{subsec:computation-of-f-check-hat}.

\subsection{Base Polytope Non-Emptiness Algorithm\label{sec:Base-Polytope-Non-Emptiness-Algorithm}}

We now discuss and analyze algorithmic approaches
to the detection of the non-emptiness of the base polytope $\basepolytope[f]$
for a crossing submodular function $f$. First, the original method
introduced by Fujishige~\cite{fujishige1984structures}  is presented. Then an alternative approach for functions on
crossing families closed under complement is presented. The time complexity of the two approaches
is finally discussed.

\subsubsection{Fujishige's Approach}

Fujishige \cite{fujishige1984structures} proved the following theorem.

\begin{theorem}
\label{thm:fujishige-non-emptiness-condition}Let $f\colon\setfamily\rightarrow\mathbb{R}$
be a crossing submodular function on the crossing family $\setfamily\subseteq\powerset N$
for which
$N  \in\setfamily$,  $\emptyset\notin\setfamily$ or $f\left(\emptyset\right)  \geq0$,  $f\left(N\right)  =0$.
Then it holds that $\basepolytope[f]\neq\emptyset$ if and only if
\begin{equation}
\begin{cases}
\sum\nolimits_{T\in\mathcal{P}}f\left(T\right)\geq0 & \text{\ensuremath{\forall\mathcal{P}\subseteq\setfamily} proper partition of \ensuremath{N}},\\
\sum\nolimits_{T\in\mathcal{Q}}f\left(T\right)\geq0 & \text{\ensuremath{\forall\mathcal{Q}\subseteq\setfamily} proper copartition of \ensuremath{N}}.
\end{cases}\label{eq:fujishige-non-emptiness-condition}
\end{equation}
\end{theorem}

\proof
Here we follow the proof presented by Schrijver in \cite[49.10]{schrijverCombinatorialOptimizationPolyhedra2004}.
If there exists $x\in\basepolytope[f]=\extendedpolymatroid[f]\cap\left\{ x\left(N\right)=0\right\} $,
then $\forall\mathcal{P}\subseteq\setfamily$ proper partition of
$N$
\[
\sum\nolimits_{T\in\mathcal{P}}f\left(T\right)\geq\sum\nolimits_{T\in\mathcal{P}}x\left(T\right)=x\left(N\right)=0,
\]
and $\forall\mathcal{Q}\subseteq\setfamily$ proper copartition of
$N$
\begin{align*}
\sum\nolimits_{T\in\mathcal{Q}}f\left(T\right)
&\geq\sum\nolimits_{T\in\mathcal{Q}}x\left(T\right)\\
&=\sum\nolimits_{T\in\mathcal{Q}}\left(x\left(N\right)-x\left(N\setminus T\right)\right)=\left|\mathcal{Q}\right|x\left(N\right)-x\left(N\right)=0,
\end{align*}

thus the condition of (\ref{eq:fujishige-non-emptiness-condition})
is necessary.

To see its sufficiency, consider $g\coloneqq\hat{f}$ on $\mathcal{D}\coloneqq\hat{\setfamily}=\left\{ \bigcap_{T\in\mathcal{T}}T\setdef\mathcal{T}\subseteq\setfamily\right\} \setminus\left\{ \emptyset\right\} $
and $h\coloneqq\check{g}$ on $\mathcal{E}\coloneqq\check{\mathcal{D}}=\left\{ \bigcup\nolimits_{T\in\mathcal{T}}T\setdef\mathcal{T}\subseteq\mathcal{D}\right\}$.
Notice how $N\in\mathcal{E}$. Proposition~\ref{prop:intersecting-submodularity-of-f-hat}
guarantees that $g$ is intersecting submodular on $\mathcal{D}$,
and by Proposition~\ref{prop:lattice-submodularity-of-f-check} $h$ is thus lattice
submodular on $\mathcal{E}$. 
We have that
\[
\basepolytope[f]=\basepolytope[g]=\extendedpolymatroid[g]\cap\left\{ x\left(N\right)=0\right\} =\extendedpolymatroid[h]\cap\left\{ x\left(N\right)=0\right\} ,
\]
We know that $\extendedpolymatroid[h]\neq\emptyset$,
therefore it remains to prove that an element $x$ of $\extendedpolymatroid[h]$
can be found for which $x\left(N\right)=0$. But this is equivalent
to asking that $h\left(N\right)\geq0$.
The condition ``$\basepolytope[f]\neq\emptyset$ if and only if (\ref{eq:fujishige-non-emptiness-condition})''
can thus be rewritten as
``
$h\left(N\right)\geq0\text{ if and only if }(\ref{eq:fujishige-non-emptiness-condition})$
''.

To prove sufficiency, we thus need to prove that $h\left(N\right)\geq0$
when (\ref{eq:fujishige-non-emptiness-condition}) holds. Assume by
contradiction that $h\left(N\right)<0$, and let $\mathcal{P}\subseteq\mathcal{D}$
be a proper partition of $N$ and $\mathcal{Q}_{S}\subseteq\mathcal{C}$
be proper copartitions of $N\setminus S$ for each $S\in\mathcal{P}$ for which
\[h\left(N\right)  =\sum\nolimits_{T\in\mathcal{P}}g\left(T\right), \hspace{100pt} g\left(S\right)  =\sum\nolimits_{T\in\mathcal{Q}_{S}}f\left(T\right)\quad\forall S\in\mathcal{P}.\]
Let $\mathcal{F}$ be the family formed by the union with repetition
of all $\mathcal{Q}_{S}$ $\forall S\in\mathcal{P}$. For $\mathcal{F}$
it holds:
\begin{enumerate}
\item \label{enu:family-F-sum-property}$\sum\nolimits_{T\in\mathcal{F}}f\left(T\right)<0$;
\item \label{enu:family-F-element-property}each element of $N$ is contained
in the same number of sets of $\mathcal{F}$.
\end{enumerate}
The property~\ref{enu:family-F-sum-property}.\ comes from the hypothesis
$h\left(N\right)<0$, while~\ref{enu:family-F-element-property}.\ by
the properties of copartitions (let $N$ be
 a set, and let $\mathcal{Q}$ be a copartition of a subset $S\subseteq N$ ; then, collectively, the sets of $\mathcal{Q}$ contain exactly $\left|\mathcal{Q}\right|-1$ copies of each element of $S$ and $\left|\mathcal{Q}\right|$ copies of each element of $N\setminus S$).
The family $\mathcal{F}$ can be successively refined by replacing
with $T'\cup T'',T'\cap T''$ any $T',T''\in\mathcal{F}$ such that
$T'\cap T''\neq\emptyset,T'\cup T''\neq N$ and $T'\nsubseteq T''\nsubseteq T'$.
By Lemma~\ref{refinement-sum-decrease} the sum $\sum_{T\in\mathcal{F}}\left|T\right|\left|N\setminus T\right|$
decreases at each step, proving that the refinement procedure eventually
terminates. At each step, the crossing submodularity of $f$ grants
that the sum $\sum\nolimits_{T\in\mathcal{F}}f\left(T\right)$ does
not increase, preserving~\ref{enu:family-F-sum-property}., while set-theoretical
considerations show that~\ref{enu:family-F-element-property}.\ is also
preserved. We can therefore assume \textsc{w.l.o.g.}~that the family
$\mathcal{F}$, with the two listed properties, is cross-free, that
is for any $T',T''\in\mathcal{F}$ it holds either $T'\subseteq T''$,
$T'\supseteq T''$, $T'\cap T''=\emptyset$ or $T'\cup T''=N$.

We now show that we can iteratively find and remove a partition or
copartition $\mathcal{R}$ of $N$ from $\mathcal{F}$, preserving
both~\ref{enu:family-F-sum-property}.\ and~\ref{enu:family-F-element-property}.,
until it becomes empty. Select $T\in\mathcal{F}$, and assume $T\neq\emptyset,N$
(otherwise we trivially conclude). Let $\mathcal{X}$ be the family
of inclusion-wise maximal sets $\mathcal{F}$ contained in $N\setminus T$,
and $\mathcal{Y}$ be the family of inclusion-wise minimal sets $\mathcal{F}$
that contain $N\setminus T$. The elements in $\mathcal{X}$ are pairwise disjoint, due to $\mathcal{F}$
being cross-free. Similarly, the complements of sets in $\mathcal{Y}$ are pairwise disjoint.

If either $\bigcup\mathcal{X}=N\setminus T$ or $\bigcap\mathcal{Y}=N\setminus T$
then $\mathcal{X}\cup\left\{ T\right\} $ and $\mathcal{Y}\cup\left\{ T\right\} $
are respectively a partition or a copartition of $N$. If instead
there exist $s\in\left(N\setminus T\right)\setminus\left(\bigcup\mathcal{X}\right)$
and $t\in T\cap\left(\bigcap\mathcal{Y}\right)$, since $s\notin T$
and $t\in T$ from~\ref{enu:family-F-element-property}.\ it follows
that there must exist $T'\in\mathcal{F}$ with $s\in T'$ and $t\notin T'$,
and therefore $T\nsubseteq T'\nsubseteq T$. This leaves either $T\cap T'=\emptyset$
(but then $T'\subseteq X$ for some $X\in\mathcal{X}$, hence $s\in\bigcup\mathcal{X}$)
or $T\cup T'=N$ (but then $T'\supseteq Y$ for some $Y\in\mathcal{Y}$,
hence $t\notin\bigcap\mathcal{Y}$), both of which lead to a contradiction.
A partition or copartition $\mathcal{R}$ of $N$ can therefore always
be found inside $\mathcal{F}$.

By the properties of partition and copartitions it follows that removing
the elements $\mathcal{R}$ from $\mathcal{F}$ preserves the validity
of~\ref{enu:family-F-element-property}., and the procedure can thus
be repeated until $\mathcal{F}=\emptyset$. But by (\ref{eq:fujishige-non-emptiness-condition}),
the removal of $\mathcal{R}$ also preserves the validity of~\ref{enu:family-F-sum-property}..
Since~\ref{enu:family-F-sum-property}.\ cannot be true for $\mathcal{F}=\emptyset$,
the assumption $h\left(N\right)<0$ is contradicted, concluding the
proof.
\endproof
The result of Theorem~\ref{thm:fujishige-non-emptiness-condition} can be extended
to a crossing submodular function $f\colon\setfamily\rightarrow\mathbb{R}$
without $f\left(N\right)=0$, by replacing it with the
set function $f'\colon\setfamily\rightarrow\mathbb{R}$ defined by
\[
f'\left(S\right)\coloneqq f\left(S\right)-{\textstyle \frac{\left|S\right|}{\left|N\right|}}f\left(N\right)\qquad\forall S\in\setfamily.
\]
Such function $f'$ is still crossing submodular on $\setfamily$
since, for any $S,T\in\setfamily$ with $S\cap T\neq\emptyset$ and
$S\cup T\neq N$, we have 
\begin{equation*}
        \begin{alignedat}{1}
f'\left(S\right)+f'\left(T\right) & =f\left(S\right)+f\left(T\right)-{\textstyle \frac{\left|S\right|+\left|T\right|}{\left|N\right|}}f\left(N\right)\\
 & \leq f\left(S\cap T\right)+f\left(S\cup T\right)-{\textstyle \frac{\left|S\cap T\right|+\left|S\cup T\right|}{\left|N\right|}}f\left(N\right)\\
 & =f'\left(S\cap T\right)+f'\left(S\cup T\right)
 \end{alignedat}
 \end{equation*}
from the crossing submodularity of $f$ and the fact that $\left|S\right|+\left|T\right|=\left|S\cap T\right|+\left|S\cup T\right|$,
and is such that $\basepolytope[f]\neq\emptyset\iff\basepolytope[f']\neq\emptyset$
since
\[
x\in\basepolytope[f]\quad\iff\quad\overline{x}\in\basepolytope[f']\qquad\text{with}\quad\overline{x}_{s}\coloneqq x_{s}-{\textstyle \frac{1}{\left|N\right|}}f\left(N\right)\quad\forall s\in N.
\]


The proof of Theorem~\ref{thm:fujishige-non-emptiness-condition} provides
a polynomial-time algorithmic procedure to detect the non emptiness
of $\basepolytope[f]$, based on the condition
\begin{equation}
\basepolytope[f]\neq\emptyset\qquad\text{if and only if}\qquad h\left(N\right)\geq0.\label{eq:hat-check-non-emptiness-condition}
\end{equation}
Indeed, it proves that it is enough to compute $h\left(N\right)$
and check for its non-negativity to obtain $\basepolytope[f]\neq\emptyset$.
The computational complexity of this procedure is discussed in \S\ref{subsec:complexity-fujishige}.

\subsubsection{Discrete Sandwich Theorem Approach\label{subsec:Discrete-Sandwich-Theorem-Approach}}

The key idea behind our new approach relies on a result,
known in literature as \emph{discrete sandwich theorem}, whose original
formulation is due to Frank \cite{frankAlgorithmSubmodularFunctions1982a}
(additional references are \cite{fujishigeSystemLinearInequalities1984},
\cite[46.2b]{schrijverCombinatorialOptimizationPolyhedra2004}). Here,
we provide it in an equivalent formulation.
\begin{lemma}[Frank's discrete sandwich theorem]
\label{lem:Franks-discrete-sandwich-theorem}Let $f\colon\setfamily\rightarrow\mathbb{R}$, $g\colon\setfamily\rightarrow\mathbb{R}$
be two lattice submodular functions on $\setfamily$, with $f\left(\emptyset\right)\geq0$ and $g\left(\emptyset\right)\geq0$ if $\emptyset\in\setfamily$. Then there exists
$x\in\mathbb{R}^{|N|}$ satisfying 
\[
-g\left(S\right)\leq x\left(S\right)\leq f\left(S\right)\qquad\forall S\in\setfamily,
\]
if and only if $-g\left(S\right)\leq f\left(S\right)$ for all $S\in\setfamily$,
or equivalently if and only if 
\[
\min_{S\in\setfamily}\left(f\left(S\right)+g\left(S\right)\right)\geq0.
\]
\end{lemma}

The final rewriting as a minimization problem is most useful, as it
allows for an efficient algorithmic check of the condition, since
the function $f+g$ is submodular, and its minimum can be computed
with efficient submodular function minimization algorithms.

All elements are now in place to state our main result.

\global\long\def\setfamilyaux{\setfamily'}%

\begin{theorem}
\label{thm:sandwich-non-emptiness-condition}Let $f\colon\setfamily\rightarrow\mathbb{R}$
be a crossing submodular function on the crossing family $\setfamily\subseteq\powerset N$
for which
$\complementary{\setfamily}  =\setfamily$; $ \emptyset,N  \in\setfamily$;  $f\left(\emptyset\right)  \geq0$.
Then for any selection of $s\in N$ and for \mbox{$\setfamilyaux\coloneqq\restrictedfamily[N\setminus\left\{ s\right\} ]=\left\{ S\in\setfamily\setdef S\subseteq\left(N\setminus\left\{ s\right\} \right)\right\} $}
it holds that
\begin{equation}
\basepolytope[f]\neq\emptyset\qquad\text{if and only if}\qquad\min_{S\in\check{\setfamilyaux}}\left(\check{f}\left(S\right)+\check{\complementary f}\left(S\right)\right)\geq0.\label{eq:sandwich-non-emptiness-condition}
\end{equation}
\end{theorem}

\proof
Notate $N\setminus\left\{ s\right\} $ as $\Ns$ for a selected element
$s\in N$. The system defining the base polytope $\basepolytope[f]$
is
\[
\begin{cases}
\phantom{x\left(N\right)}\mathllap{x\left(S\right)}\leq f\left(S\right) & \forall S\in\setfamily,\\
x\left(N\right)=f\left(N\right),
\end{cases}
\]
and can rewritten by explicitly writing the component $x_{s}$ of
the variable $x$ as
\begin{equation}
\begin{cases}
\phantom{x\left(\Ns\right)+x_{s}}\mathllap{x\left(S\right)}\leq f\left(S\right) & \forall S\subseteq\Ns\text{ s.t. \ensuremath{S\in\setfamily}},\\
\phantom{x\left(\Ns\right)+x_{s}}\mathllap{x\left(S\right)+x_{s}}\leq f\left(S\cup\left\{ s\right\} \right) & \forall S\subseteq\Ns\text{ s.t. \ensuremath{\left(S\cup\left\{ s\right\} \right)\in\setfamily}},\\
x\left(\Ns\right)+x_{s}=f\left(N\right).
\end{cases}\label{eq:base-polytope-rewrite-1}
\end{equation}
By subtracting the last two constraints we obtain
\[
\left(x\left(\Ns\right)+\bcancel{x_{s}}\right)-\left(x\left(S\right)+\bcancel{x_{s}}\right)=x\left(\Ns\setminus S\right)\geq f\left(N\right)-f\left(S\cup\left\{ s\right\} \right),
\]
for each $S\subseteq\Ns$ such that $\left(S\cup\left\{ s\right\} \right)\in\setfamily$.
By now taking $S'\coloneqq\Ns\setminus S$, for which consequently
$S\cup\left\{ s\right\} =\left(\Ns\setminus S'\right)\cup\left\{ s\right\} =N\setminus S'$ and thus $S'\in\complementary{\setfamily}$,
the expression above becomes
\begin{equation}
x\left(S'\right)\geq f\left(N\right)-f\left(N\setminus S'\right)=-\complementary f\left(S'\right)\qquad\forall S'\subseteq\Ns\text{ s.t. \ensuremath{S'\in\complementary{\setfamily}}},\label{eq:base-polytope-complementary-expr}
\end{equation}
with $\complementary f$ again crossing submodular. Replacing the
second constraint of (\ref{eq:base-polytope-rewrite-1}) with (\ref{eq:base-polytope-complementary-expr}),
and recalling that $\complementary{\setfamily}=\setfamily$ by hypothesis,
we have that $x\in\basepolytope[f]$ if and only if
\begin{equation}
\begin{cases}
\phantom{x\left(\Ns\right)+x_{s}}\mathllap{x\left(S\right)}\leq f\left(S\right) & \forall S\subseteq\Ns\text{ s.t. \ensuremath{S\in\setfamily}},\\
\phantom{x\left(\Ns\right)+x_{s}}\mathllap{x\left(S\right)}\geq-\complementary f\left(S\right) & \forall S\subseteq\Ns\text{ s.t. \ensuremath{S\in\setfamily}},\\
x\left(\Ns\right)+x_{s}=f\left(N\right).
\end{cases}\label{eq:base-polytope-rewrite-2}
\end{equation}

Since no inequalities of the system (\ref{eq:base-polytope-rewrite-2})
are relative the component $x_{s}$, this is uniquely determined by
$x_{s}=f\left(N\right)-x\left(\Ns\right)$ for any instantiation of
all other components of $x$. The feasibility of the system (\ref{eq:base-polytope-rewrite-2})
(and thus the non-emptiness of $\basepolytope[f]$) is therefore equivalent
to the existence of a vector $x$ in the subspace $\mathbb{R}^{|\Ns|}$
of $\mathbb{R}^{|N|}$ for which
\begin{equation}
-\complementary f\left(S\right)\leq x\left(S\right)\leq f\left(S\right)\qquad\forall S\subseteq\Ns\text{ s.t. \ensuremath{S\in\setfamily}}.\label{eq:non-emptiness-condition-crossing-submodular}
\end{equation}

Replacing the functions $f$ and $\complementary f$ with their restrictions
$\restrictedfun[N_{s}]\coloneqq\restrict[][\setfamilyaux]$ and $\restrictedfun[N_{s}][\complementary f]\coloneqq\restrict[\complementary f][\setfamilyaux]$
to the same domain $\setfamilyaux=\restrictedfamily[\Ns]=\setfamily\cap\powerset{\Ns}=\complementary{\setfamily}\cap\powerset{\Ns}$
trivially does not change equation (\ref{eq:non-emptiness-condition-crossing-submodular}). The expression thus becomes equivalent to asking for 
\begin{equation}
-\check{\restrictedfun[N_{s}][\complementary f]}\left(S\right)\leq x\left(S\right)\leq\check{\restrictedfun[N_{s}]}\left(S\right)\qquad\forall S\in\check{\setfamilyaux}.\label{eq:non-emptiness-condition-f-check}
\end{equation}
Both $\restrictedfun[N_{s}]$ and $\restrictedfun[N_{s}][\complementary f]$
are intersecting submodular,
with the values of $\check{\restrictedfun[N_{s}]}\left(S\right)$
and $\check{\restrictedfun[N_{s}][\complementary f]}\left(S\right)$
computable though Algorithm~\ref{alg:f-check-computation}. By Proposition~\ref{prop:lattice-submodularity-of-f-check}
the functions $\check{\restrictedfun[N_{s}]}$ and $\check{\restrictedfun[N_{s}][\complementary f]}$
are lattice submodular on $\check{\setfamilyaux}$, and Frank's discrete
sandwich theorem (Lemma~\ref{lem:Franks-discrete-sandwich-theorem}) therefore
implies the equivalence between (\ref{eq:non-emptiness-condition-f-check})
and the minimization problem
\[
\min_{S\in\check{\setfamilyaux}}\left(\check{\restrictedfun[N_{s}]}\left(S\right)+\check{\restrictedfun[N_{s}][\complementary f]}\left(S\right)\right)\geq0.
\]
But since from the definition of the check operator it holds $\check{\restrictedfun[N_{s}]}\left(S\right)=\check{f}\left(S\right)$
for all $S\in\check{\setfamilyaux}$ (since always $S\neq N$), and similarly for $\complementary f$,
we conclude that
\[
\basepolytope[f]\neq\emptyset\qquad\text{if and only if}\qquad\min_{S\in\check{\setfamilyaux}}\left(\check{f}\left(S\right)+\check{\complementary f}\left(S\right)\right)\geq0.
\]
\endproof

\subsubsection{Complexity}

This section contains a discussion on the computational complexity
of the presented algorithms. Such complexities are expressed in terms
of the two subroutines of function evaluation and submodular function
minimization. Let therefore $\evaltime[f]$ be the time-complexity
of evaluating a set function $f\colon\setfamily\rightarrow\mathbb{R}$
over any subset in $\setfamily\subseteq\powerset N$ with $n\coloneqq\left|N\right|$,
and when $f$ is lattice submodular on $\setfamily$ let $\subrettime[f][n]$
be the complexity of an oracle that retrieves both the minimum of
$f$ and a subset $S\in \setfamily$ that minimizes it.
When comparing approaches, we will assume 
\[
\subrettime[f][n]\coloneqq\bigo(\complvarp\cdot\evaltime[f]+\complvarq),
\]
with $\complvarp$ and $\complvarq$ functions of $n$ (asymptotically)
bounded by a polynomial in $n$.

The survey \cite{leeFasterCuttingPlane2015} states that, in the case
where $\setfamily=\powerset N$, the best-known method for minimizing
the submodular function $f$ has a complexity of \\\mbox{$\bigo(n^{3}\log^{2}n\cdot\evaltime[f]+n^{4}\log^{\bigo(1)}n)$}.
Such method does not directly retrieve a minimizing subset $S\subseteq N$,
which can however be computed with $n$ calls to the minimization
procedure. The values for $\complvarp$ and $\complvarq$ when $\setfamily=\powerset N$
could therefore be chosen as $n^{4}\log^{2}n$ and $n^{5}\log^{\bigo(1)}n$
respectively. Notice in particular how in this case $\subrettime[f][n]$
can be assumed equal to $\bigo(\complvarp\cdot\evaltime[f])$ any
time $\evaltime[f]$ is $\bigomega\left(n\right)$.
In general, it is reasonable to assume that evaluating $f\left(\cdot\right)$
is at least as hard as evaluating $x\left(\cdot\right)$ for any $S\subseteq N$,
which has a complexity of $\bigo(n)$ when the complexity of retrieving
$x_{s}$ for any $s\in N$ is taken to be $\bigo(1)$.

The minimization of a general lattice submodular function $f\colon\setfamily\rightarrow\mathbb{R}$
on the lattice family $\setfamily$ can be achieved through a polynomial-time
reduction to the submodular function minimization procedure, as described in \cite[§49.3]{schrijverCombinatorialOptimizationPolyhedra2004}.

\subsubsection*{Computation of $\check{f}$ and $\hat{f}$\label{subsec:computation-of-f-check-hat}}

To begin, notice how the complexity $\evaltime[\complementary f]$
of evaluating $\complementary f$ is the same as $\evaltime[f]$,
since from its definition computing $\complementary f$ requires two
evaluations of $f$ (of which one always to $f\left(N\right)$), thus
$\evaltime[\complementary f]=\bigo(\evaltime[f])$.

We now follow Algorithm~\ref{alg:f-check-computation}. The ordering of line~\ref{line:f-check-init-ordering}
is arbitrary and requires no computation.
The proof of Proposition~\ref{prop:f-check-computation} uses the bound $M$ only through
$-M\left|U\right|\leq f\left(U\right)$ for $U\in\setfamily$ with $\emptyset\neq U\subseteq S$,
so $M$ can be replaced by $\max\left\{ 0,-\min\left\{ f\left(U\right)\setdef U\in\setfamily,\emptyset\neq U\subseteq S\right\} \right\} $,
which is obtained by minimizing $f$ over the lattice family $\left\{ U\in\setfamily\setdef s\in U\subseteq S\right\} $
for each $s\in S$, that is with $k$ calls to the minimization oracle.

Computing the function $\check{f}$ for an admissible input function
$f$ requires first the iterative retrieval for $i=1,...,k$ (lines~\ref{line:f-check-construction-1}-\ref{line:f-check-construction-2})
of the minimizer of the function $\phi$ defined by $\phi\colon U\mapsto f\left(U\cup\left\{ s_{i}\right\} \right)-x\left(U\right)$
for $U\in\powerset{S\setminus\left\{ s_{i}\right\} }$ with $U\cup\left\{ s_{i}\right\} \in\setfamily$,
which is lattice submodular since $x$ is trivially modular and the
intersecting submodularity of $f$ implies the lattice submodularity
of $U\mapsto f\left(U\cup\left\{ s_{i}\right\} \right)$. By assumption
we can neglect the evaluation of $x\left(\cdot\right)$ compared to
that of $f\left(\cdot\right)$, thus $\evaltime[\phi]=\bigo(\evaltime[f])$
and consequently $\subrettime[\phi][k]=\bigo(\subrettime[f][k])$.

Since $\mathcal{T}$ contains at most $k$ elements, a naïve implementation
of its refinement iteration --- cycling over all pairs of elements
in $\mathcal{T}$ and keeping track of all sets to be merged ---
would terminate in at most $k^{2}$ steps. The final complexity of
the computation is therefore equal to
\[
\checktime[f][n]\coloneqq  \bigo(n\cdot\subrettime[\phi][k]+n^{2})
=  \bigo(n\cdot\subrettime[f][k]).
\]
We already discussed  how to compute $\hat{f}$
from $\check{(\complementary f)}$ for an admissible input $f$ in
the assumption that $f\left(N\right)=0$ with Algorithm~\ref{alg:f-check-computation}. The time-complexity of computing
$\hat{f}\left(S\right)$ for a given $S$ is hence the same as that
of computing $\check{(\complementary f)}\left(N\setminus S\right)$,
that is $\checktime[\complementary f][k]=\checktime[f][k]$. We will
therefore refer to it directly as $\checktime[f][k]$.

\subsubsection*{Fujishige\textquoteright s Approach\label{subsec:complexity-fujishige}}

Notate $g\coloneqq\hat{f}$ and $h\coloneqq\check{g}$ as in the proof
of Theorem~\ref{thm:fujishige-non-emptiness-condition}. Since $f$ is crossing
submodular, the input requirements are not directly satisfied in order
to compute $g$ with Algorithm~\ref{alg:f-check-computation}. However, we notice
that the value of $g$ is accessed only during the computation of
$h$ (with Algorithm~\ref{alg:f-check-computation} without issues this time,
since $g$ is intersecting submodular) which implies that $g\left(S\right)$
is never evaluated for $S=\emptyset$. We can thus again use Algorithm~\ref{alg:f-check-computation} to compute
$g$.

Given and evaluation oracle for $f$ with complexity $\oracletime[f]$, we have that the complexity of checking for the non-emptiness of $\basepolytope[f]$
through the condition (\ref{eq:hat-check-non-emptiness-condition})
is equal to that of computing $h\left(N\right)$, thus
\begin{equation}\label{eq:complexity-fujishige}
    \begin{alignedat}{1}
    \evaltime[h] & =\checktime[g][n]
  =\bigo(n\cdot\subrettime[g][n])
  =\bigo(n\complvarp\cdot\evaltime[g]+n\complvarq)
  =\bigo(n\complvarp\cdot\checktime[f][n]+n\complvarq)\\ 
  &=\bigo(n^{2}\complvarp\cdot\subrettime[f][n]+n\complvarq) 
  =\bigo(n^{2}\complvarp^{2}\cdot\oracletime[f]+n^{2}\complvarp\complvarq+\cancel{n\complvarq}) \\
  &=\bigo(n^{2}\complvarp^{2}\cdot\oracletime[f]+n^{2}\complvarp\complvarq).
    \end{alignedat}
\end{equation}

\subsubsection*{Discrete Sandwich Theorem Approach\label{subsec:complexity-sandwich}}

Consider now the test for the non-emptiness of $\basepolytope[f]$
based on the condition (\ref{eq:sandwich-non-emptiness-condition})
of Theorem~\ref{thm:sandwich-non-emptiness-condition}. This procedure relies
on the evaluation of the sum of the functions of $\check{f}$ and
$\check{\complementary f}$, where $f$ (and thus $\complementary f$)
is crossing submodular. Again, Algorithm~\ref{alg:f-check-computation} can be
used nonetheless to compute the two functions $\check{f},\check{\complementary f}$ as they are evaluated
on sets $S\subseteq N\setminus\left\{ s\right\} $ (or equivalently the intersecting submodular restrictions $\restrictedfun[N_{s}]$
and $\restrictedfun[N_{s}][\complementary f]$ are employed in computing
$\check{\restrictedfun[N_{s}]}$ and $\check{\restrictedfun[N_{s}][\complementary f]}$,
as in the proof of Theorem~\ref{thm:sandwich-non-emptiness-condition}). Hence, the complexity of this evaluation is 
\[
\evaltime[\left(\check{f}+\check{\complementary f}\right)]=\bigo(\checktime[f][n-1]+\checktime[\complementary f][n-1])=\bigo(\checktime[f][n])=\bigo(n\cdot\subrettime[f][n]),
\]
since $\evaltime[\complementary f]=\bigo\left(\oracletime[f]\right)$.
Again as a function of the complexity $\oracletime[f]$ of an evaluation oracle for $f$, the total complexity of the non-emptiness test is that of the minimization
of the submodular function $\check{f}+\check{\complementary f}$,
that is 
\begin{equation*}
    \begin{alignedat}{1}
      \subrettime[\left(\check{f}+\check{\complementary f}\right)][n] & =\bigo(\complvarp\cdot\evaltime[\left(\check{f}+\check{\complementary f}\right)]+\complvarq)\nonumber \\
 & =\bigo(n\complvarp\cdot\subrettime[f][n]+\complvarq)\nonumber \\
 & =\bigo(n\complvarp^{2}\cdot\oracletime[f]+n\complvarp\complvarq+\cancel{\complvarq})\nonumber \\
 & =\bigo(n\complvarp^{2}\cdot\oracletime[f]+n\complvarp\complvarq). 
    \end{alignedat}
\end{equation*}
We therefore observe an improvement by a factor of $n$ from the complexity
of (\ref{eq:complexity-fujishige}).




%
%

\bibliographystyle{spmpsci}      
\bibliography{bib_mathOR}   

\end{document}